\documentclass[a4paper]{coneural}
\usepackage{amssymb}
\usepackage{amsmath}
\usepackage{amsthm}
\usepackage{authblk}
\usepackage[vlined]{algorithm2e}
\hypersetup{
    pdfcreator={RIST},
    pdfproducer={RIST}
}

\title{A new class of metrics for spike trains}

\author[1,2,3]{C\u{a}t\u{a}lin V. Rusu}
\author[1]{R\u{a}zvan V. Florian}

\affil[1]{Center for Cognitive and Neural Studies (Coneural), Romanian Institute of Science and
Technology, Cluj-Napoca, Romania
}
\affil[2]{Computer Science Department, Babe\c{s}-Bolyai University,
 Cluj-Napoca, Romania
}				
\affil[3]{Frankfurt Institute for Advanced Studies, Frankfurt am Main, Germany
}
\date{}

\newtheorem{proposition}{Proposition}[section]
\newtheorem{lemma}{Lemma}[section]

\begin{document}

\maketitle

\hyphenation{spi-kes Pom-pe-iu pse-u-do-co-de nu-me-ri-cal-ly}

\begin{abstract}
The distance between a pair of spike trains, quantifying the differences between them, can be measured using various metrics. Here we introduce a new class of spike train metrics, inspired by the Pompeiu-Hausdorff distance, and compare them with existing metrics. Some of our new metrics (the modulus-metric and the max-metric) have characteristics that are qualitatively different than those of classical metrics like the van Rossum distance or the Victor \& Purpura distance. The modulus-metric and the max-metric are particularly suitable for measuring distances between spike trains where information is encoded in bursts, but the number and the timing of spikes inside a burst does not carry information. The modulus-metric does not depend on any parameters and can be computed using a fast algorithm, in a time that depends linearly on the number of spikes in the two spike trains. We also introduce localized versions of the new metrics, which could have the biologically-relevant interpretation of measuring the differences between spike trains as they are perceived at a particular moment in time by a neuron receiving these spike trains.
\\
\\
An edited version has been published in Neural Computation 26(2), pp. 306--348, 2014, \url{http://dx.doi.org/10.1162/NECO_a_00545}
\\
\\
Rate or review this paper at \url{http://epistemio.com/p/r43uM8ke}
\end{abstract}

\newpage

\section{Introduction}
In recent years several spike train distances, some inspired by existing mathematical distances while others not, have been proposed and used to measure the variability of the neural activity \cite{VP96,VP97,R01,schreiber2003,V05,schrauwen07,kreuz2007,kreuz2009,naud2010,kreuz2011,kreuz2013}. The distance between two spike trains reflects their similarity or dissimilarity. Because how information is represented in the spatio-temporal patterns of spike times exchanged by neurons is still a heavily debated topic in neuroscience, metrics based on different neural codes are available. Traditionally it was thought that the mean firing rate of neurons encapsulated all the relevant information exchanged by neurons. This idea dates back to the work of \citet{A26} who showed that the firing rate of motor neurons is proportional to the force applied. More recently, scientists have revealed increasing evidence of the importance of precise spike timings to representing information in the brain \cite{B04,vanrullen2005,TFS08}. For example, temporally-structured multicell spiking patterns were observed in the hippocampus and cortex, and were associated to memory traces \cite{Neal99,JW07} while the coding of information in the phases of spikes relative to a background oscillation has been observed in many brain regions \cite{LSLR05,JKEF07,FNS07,Meal08,SWM09}. This change in viewpoints from rate codes to spike-time codes is also reflected in spike train metrics.

The most basic metrics are the ones that rely on counting the total number of spikes within a spike train. A major drawback of such an approach is that all the temporal structure is lost. Even though binning techniques were introduced as a way to overcome this loss by dividing the spike train into discrete bins, the temporally--encoded information within a bin was neglected \cite{GASS91}. Other, more complex spike train metrics can be obtained by focusing on the precise spike timing instead of their total count. One example is the \citet{R01} distance which is calculated by filtering the time series corresponding to the raw spike train with a smoothing kernel, typically an exponential one, and then using the standard Euclidean distance. Another metric is the \citet{schreiber2003} correlation-based measure, which uses a symmetric (Gaussian) filter. In both cases, the choice of the kernel's parameters is arbitrary and has a high influence on the properties of the metric.

Another metric was introduced by \citet{VP96,VP97}. According to this metric, the distance between two spike trains is given by the minimum cost of basic operations needed to transform one spike train into the other. The basic operations are insertion or deletion of spikes, with a cost of 1, and the shifting of a spike, with a cost of $q|\delta t|$ where $q$ is a parameter and $\delta t$ the shifting interval. The parameter $q$ significantly influences the behavior of the metric: for $q = 0$ the metric counts the difference in the total number of spikes, while for large values of $q$ the metric returns the number of non-coincident spikes.

\citet{kreuz2007,kreuz2009,kreuz2011,kreuz2013} have introduced more recently a series of parameter-free and time-scale independent measures of spike train synchrony.

Here we introduce a new class of spike train metrics inspired by the Pompeiu-Hausdorff distance between two non-empty compact sets \cite{Pompeiu,H14}. Preliminary results have been presented in abstract form in \cite{rusu2010}.

\section{A new class of spike train metrics}\label{sec2}
We consider bounded, nonempty spike trains of the form
\begin{equation}
T = \{t^{(1)}, \ldots , t^{(n)}\},
\end{equation}
where $t^{(i)} \in \mathbb{R}$ are the ordered spike times and  $n \in \mathbb{N}^*$ is the number of spikes in the spike train. We consider spike trains with no overlapping spikes. If $n>1$, then $t^{(i-1)} < t^{(i)}$, $\forall i>1$. We denote by $a$ and $b$ some bounds of the considered spike trains, i.e. $a \leq t^{(i)} \leq b, \; \forall  t^{(i)}$, with $a, b \in \mathbb{R}$, finite, and $a<b$. We denote by $\mathcal{S}_{[a, b]}$ the set of all possible such spike trains bounded by $a$ and $b$. We study metrics that compute the distances between two spike trains $T$ and $\bar{T}$ from $\mathcal{S}_{[a, b]}$.

The new metrics that we introduce are inspired by the Pompeiu-Hausdorff distance \cite{Pompeiu,H14}. When applied to a pair of spike trains, the Pompeiu-Hausdorff distance $h$ returns the largest difference, in absolute value, between the timing of a spike in one train and the timing of the closest spike in the other spike train:
\begin{equation}
h(T,\bar{T}) = \max\left\{\sup_{t \in T} \; \inf_{\bar{t} \in \bar{T}} \; |t-\bar{t}|, \; \sup_{\bar{t} \in \bar{T}} \; \inf_{t \in T} \; |t-\bar{t}|\right\}, \label{equation_hausdorffi}
\end{equation}
or, equivalently, the minimal number $\epsilon \geq 0$ such that the closed $\epsilon$-neighborhood of $T$ includes $\bar{T}$ and the closed $\epsilon$-neighborhood of $\bar{T}$ includes $T$:
\begin{equation}
h(T,\bar{T}) = \inf\left\{\epsilon  \; \text{such that} \; |t-\bar{t}| \leq \epsilon, \; \forall t \in T, \; \forall \bar{t} \in \bar{T} \right\}.
\end{equation}
Another equivalent form of the Pompeiu-Hausdorff distance is the following (\citealp[pp. 105--110]{P05}; \citealp[pp. 117--118]{RW09}; \citealp[pp. 47--48]{DD09}):
\begin{equation}
h(T,\bar{T}) = \sup_{x \in \mathbb{R}} \left|
\inf_{t \in T} \; |t-x| - \inf_{\bar{t} \in \bar{T}} \; |\bar{t}-x|\right|. \label{equation_hausdorff_modulus}
\end{equation}

We introduce a distance $d$ between an arbitrary timing $x \in \mathbb{R}$ and a spike train $T$:
\begin{equation}\label{equation_distance_d}
d(x,T)  = \inf_{t \in T} \; |t-x|.
\end{equation}
Eq. \ref{equation_hausdorffi} can then be rewritten as
\begin{equation}
h(T,\bar{T}) = \max\left\{\sup_{t \in T} \; d(t, \bar{T}), \; \sup_{\bar{t} \in \bar{T}} \; d(\bar{t}, T)\right\} \label{equation_hausdorffii}
\end{equation}
and Eq. \ref{equation_hausdorff_modulus} as
\begin{equation} \label{equation_hausdorff}
h(T,\bar{T}) = \sup_{x \in \mathbb{R}}  \left|
d(x, T)-d(x, \bar{T})
\right|.
\end{equation}
Because the global supremum is achieved on the interval $[a,b]$ (Appendix \ref{section_equiv_haus}), we also have:
\begin{equation}
h(T,\bar{T}) = \sup_{x \in [a, \, b]}  \left|
d(x, T)-d(x, \bar{T})
\right|. \label{equation_hausdorff_ab}
\end{equation}
The Pompeiu-Hausdorff metric has a quite poor discriminating power, as for many variations of the spike trains the distances will be equal and any spike train space endowed with this metric would be highly clustered. Our new metrics generalize the form of the Pompeiu-Hausdorff distance given in Eq. \ref{equation_hausdorff_ab}, by introducing features that are more sensitive to spike timings.

\subsection{The max-metric}
We consider $\mathbb{B}$ to be the space of all continuous, positive functions $\mathcal{H} \colon [0, b-a] \rightarrow \mathbb{R}^+$ that satisfy the condition, $\forall u \in [a, b]$,
\begin{equation}
\int_a^b \mathcal{H}(|u-s|) \; \textrm{d}s >0.
\label{equation_h_condition_integral}
\end{equation}
Because
\begin{align}
\int_a^b \mathcal{H}(|u-s|) \; \textrm{d}s & = \int_a^u \mathcal{H}(u-s) \; \textrm{d}s +  \int_u^b \mathcal{H}(s-u) \; \textrm{d}s\\
& = \int_0^{u-a} \mathcal{H}(s) \; \textrm{d}s +  \int_0^{b-u} \mathcal{H}(s) \; \textrm{d}s,
\end{align}
and because $\mathcal{H}$ is continuous, a sufficient condition for satisfying Eq. \ref{equation_h_condition_integral} is that $\mathcal{H}(0)>0$.

On compact sets, continuous functions are bounded \cite[p. 56]{Protter}. We denote by $m$ the upper bound of $\mathcal{H}$ on the interval $[0, b-a]$, i.e.
\begin{equation}
0 \leq \mathcal{H}(x) < m < \infty, \; \forall \; x \in [0, b-a].\label{equation_h_conditions}
\end{equation}
In typical applications, $\mathcal{H}(x)$ has a maximum for $x=0$ and is a decreasing function of $x$, for example an exponential,
\begin{equation}
\mathcal{H}_E(x)=\frac{1}{\tau} \; \exp\left(-\frac{x}{\tau}\right),\label{equation_H_exponential}
\end{equation}
or a Gaussian,
\begin{equation}
\mathcal{H}_G(x)=\frac{1}{\tau \; \sqrt{2\; \pi}} \; \exp\left(-\frac{x^2}{2 \; \tau^2}\right),\label{equation_H_gaussian}
\end{equation}
with $\tau$ a positive parameter.

We introduce the max-metric as
\begin{equation}\label{maxm}
d_m(T,\bar{T}) = \int_a^b \sup_{x \in [a,b]} \left\lbrace | d(x,T) - d(x,\bar{T})| \; \mathcal{H}(|s-x|)  \right\rbrace \text{d} s.
\end{equation}
The max-metric integrates, through the variation of $s$ along the interval $[a,b]$ that contains the two spike trains, the maximum of the difference, in absolute value, between the distances from a point $x$ in that interval to the two spike trains, weighted by the kernel $\mathcal{H}(|s-x|)$ which focuses locally around $s$. Figure \ref{fig:maxmetric} shows how the distance between two spike trains is computed using the max-metric.

The max-metric is a generalization of the Pompeiu-Hausdorff distance, since in the particular case that $\mathcal{H}(\cdot)=1/(b-a)$ we have $d_m(T,\bar{T})=h(T,\bar{T})$.

In Appendix \ref{section_analysis_max_metric} we show that $d_m$ is finite and that it satisfies the properties of a metric. We also show that regardless of the kernel $\mathcal{H}$ all the max-metrics are topologically equivalent to each other \cite[p. 229]{osearcoid} because they are equivalent to the Pompeiu-Hausdorff distance. Each metric will generate the same topology and thus any topological property is invariant under an homeomorphism. This means that the metrics generate the same convergent sequences in the space of spike trains $\mathcal{S}_{[a, b]}$. Thus, learning rules derived from these metrics will converge in the same way, regardless of the choice of $\mathcal{H}$.

\begin{figure}[hp]
\hfill
\begin{center}
\includegraphics[width=\textwidth]{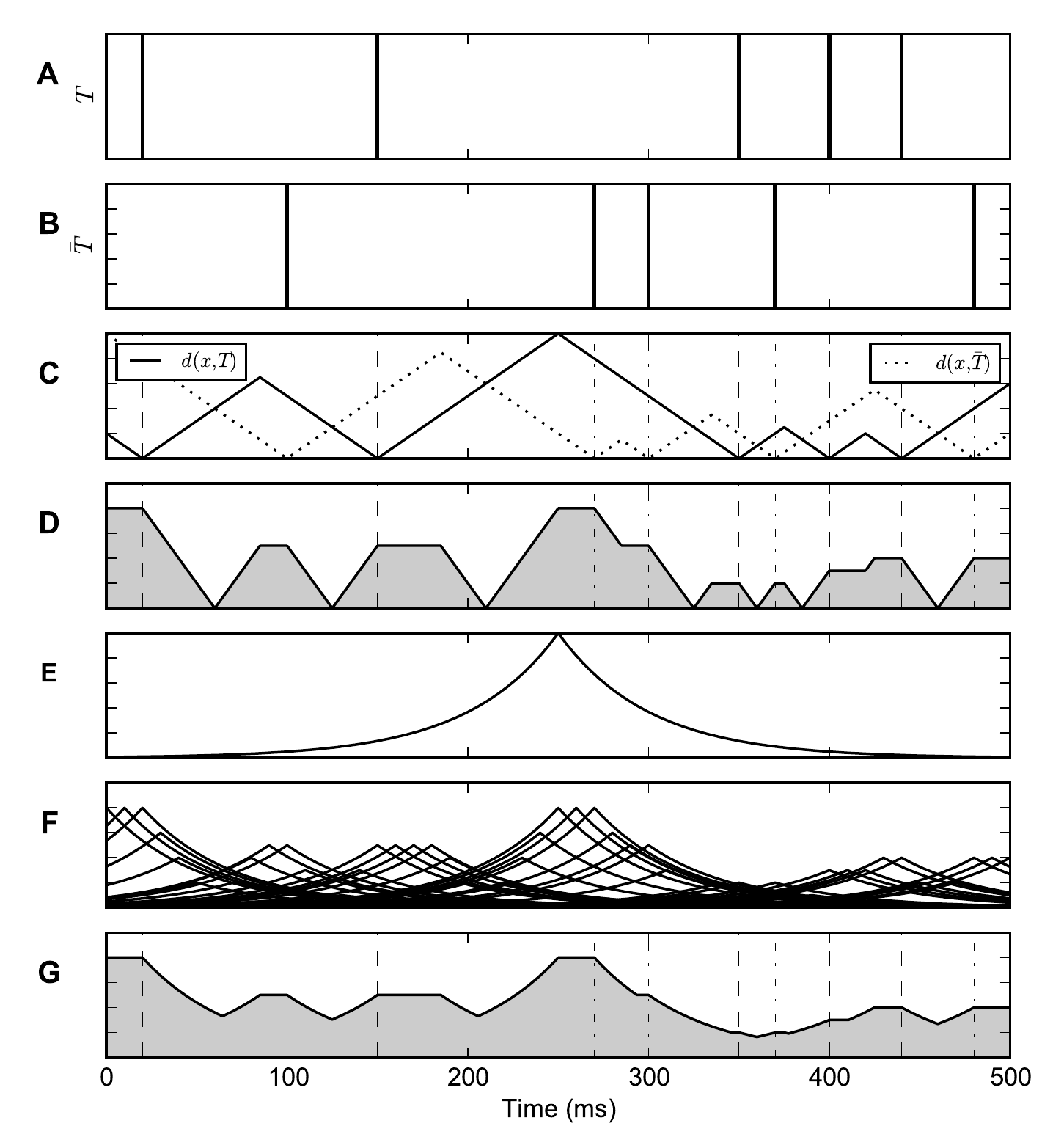}
\end{center}

\caption{The modulus-metric and the max-metric. (A) Spike train $T = \{ 20, 150, 350, 400, 440\}$ ms. Each spike time is represented as a vertical bar. (B) Spike train $\bar{T} = \{ 100, 270, 300, 370, 480\}$ ms. (C) The  distances between a timing $x$ and the spike trains, $d(x,T)$ and $d(x,\bar{T})$, as a function of $x$. (D) The difference $ \left| d(x,T) - d(x,\bar{T}) \right|$ as a function of $x$. The modulus-metric distance $d_o(T, \; \bar{T})$ is the area under this curve. (E) The kernel $\mathcal{H}(|s-x|)$ as a function of $s$ with a fixed $x=250$ ms. $\mathcal{H}$ is an exponential (Eq. \ref{equation_H_exponential}) with a decay constant $\tau=50$ ms. (F) The weighted difference $ \left| d(x,T) - d(x,\bar{T}) \right| \mathcal{H}(|s-x|)$ as function of $s$ for discrete values of $x$. (G) The supremum of the  weighted difference, $\sup_{x \in [a,b]} \left\lbrace \left| d(x,T) - d(x,\bar{T}) \right| \mathcal{H}(|s-x|)\right\rbrace$, as a function of $s$. The max-metric distance $d_m(T, \; \bar{T})$ is the area under this curve. In (C), (D), and (G), the dashed vertical lines represent the timing of spikes in $T$ and $\bar{T}$.}
\label{fig:maxmetric}
\end{figure}

\begin{figure}[h!p]
\hfill
\begin{center}
\includegraphics[width=\textwidth]{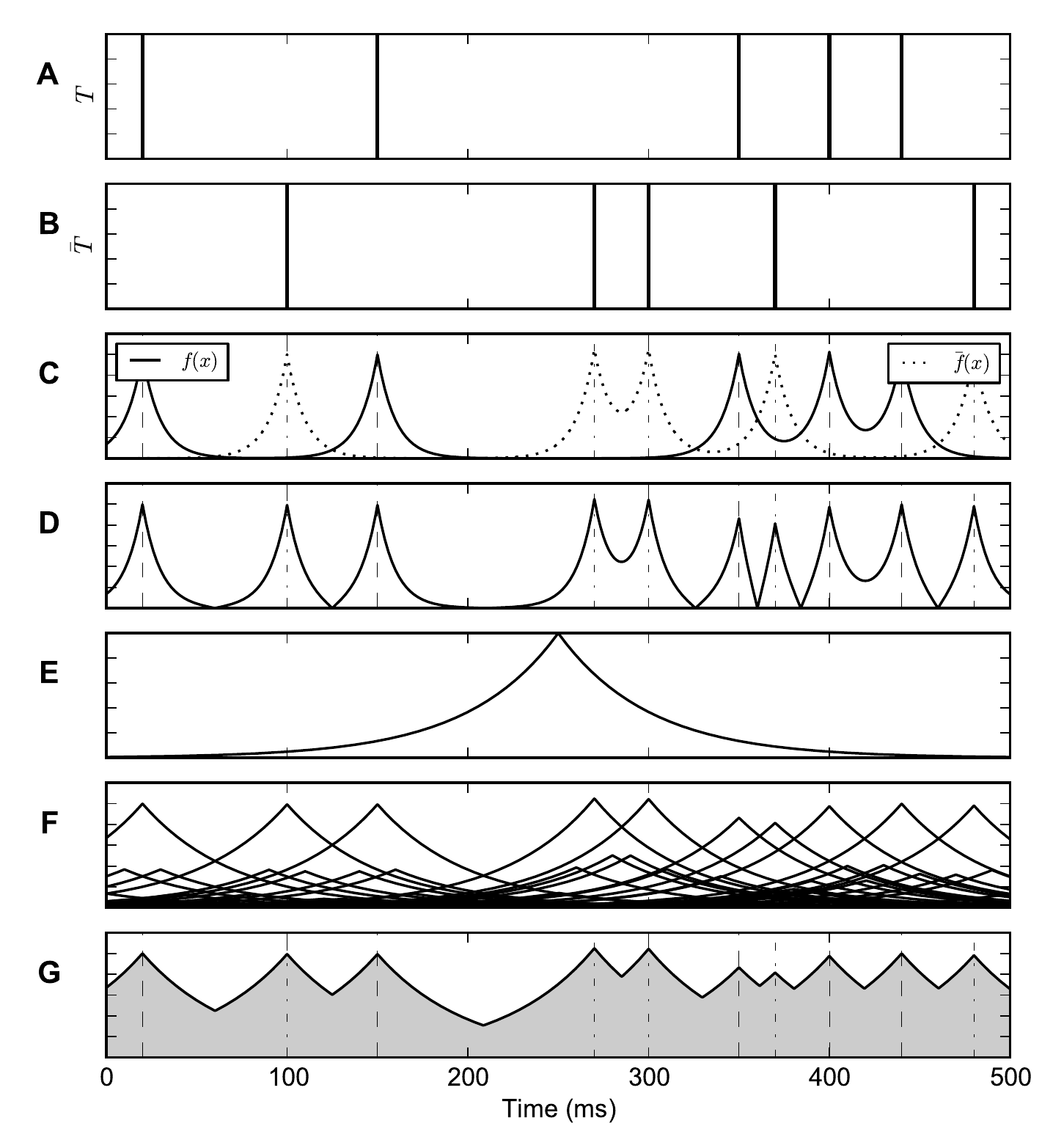}
\end{center}
\caption{The convolution max-metric. (A) Spike train $T = \{ 20, 150, 350, 400, 440\}$ ms. Each spike time is represented as a vertical bar. (B) Spike train $\bar{T} = \{ 100, 270, 300, 370, 480\}$ ms. (C) The spike trains $T$ and $\bar{T}$ filtered with an exponential kernel (Eq. \ref{equation_K_exponential}) with a  decay constant $\tau=10$ ms. (D) The difference $ \left| f(x) -  \bar{f}(x) \right|$ as a function of $x$. (E) The kernel $\mathcal{H}(|s-x|)$ as a function of $s$ with a fixed $x=250$ ms. $\mathcal{H}$ is an exponential (Eq. \ref{equation_H_exponential}) with a decay constant $\tau=50$ ms. (F) The weighted difference $ \left| f(x) - \bar{f}(x) \right| \mathcal{H}(|s-x|)$ as function of $s$ for discrete values of $x$. (G) The supremum of the weighted difference, $\sup_{x \in [a,b]} \left\lbrace \left| f(x) - \bar{f}(x) \right| \mathcal{H}(|s-x|)\right\rbrace$. The convolution max-metric distance $d_c(T, \; \bar{T})$ is the area under this curve. In (C), (D), and (G), the dashed vertical lines represent the timing of spikes in $T$ and $\bar{T}$.}
\label{fig:maxmetricconv}
\end{figure}

\subsection{The modulus-metric}
We define the modulus-metric as
\begin{equation}\label{m3p}
d_o(T,\bar{T}) = \int_{a}^{b} |d(s,T) - d(s,\bar{T})|  \; \textrm{d}s.
\end{equation}
Figure \ref{fig:maxmetric} A--D shows how the distance between two spike trains is computed using the modulus-metric. The modulus-metric is a particular case of the max-metric in the limit that $\mathcal{H}$ is
\begin{equation}
\mathcal{H}(x)=\begin{cases}
   1, &\text{if}\quad x=0,\\
    0,&\text{otherwise}.
    \end{cases}
\end{equation}

The modulus-metric does not depend on any kernels or parameters and it also allows a fast computer implementation in linear complexity. This is because the graph of the function $\phi(s)=|d(s,T) - d(s, \bar{T})|$ is made out of line segments that join or end in the following points: all timings of spikes in the two spike trains $T$ and $\bar{T}$; the time moments that lie at the middle of the interval between two neighboring spikes from the same spike train; the time moments that lie at the middle of the interval between a pair of neighboring spikes where the two spikes belong to different spike trains; and the bounds $a$ and $b$. This is exemplified in Figures \ref{fig:maxmetric} D and \ref{fig:localmaxmetric} D. We denote by $\mathcal{P}$ the set of these points. In order to compute the integral of this function $d_o=\int_a^b \phi(s) \; \text{d}s$, it is sufficient to compute the function at the points from $\mathcal{P}$. Since between these points the function is linear, the integral can be then computed exactly.

Algorithm \ref{algorithm1} presents an implementation of the $d_o$ metric in pseudo-code. In this algorithm, the function $\phi$ is computed in a set of points that includes $\mathcal{P}$ but also other points. In a second Algorithm \ref{algorithm2}, the set $\mathcal{P}$ as well as the value of $\phi$ in the points of $\mathcal{P}$ is computed with a single pass through the spikes in the two spike trains. The algorithms' duration depends linearly on the number of spikes in the two spike trains, $n+\bar{n}$. Implementations in Python and C++ of the two algorithms are freely available at \url{https://github.com/modulus-metric/}.

It can be shown that the distance $d_{o}$ is finite and that it satisfies the properties of a metric by particularizing the proofs in Appendix \ref{section_analysis_localized_modulus_metric} with $\mathcal{L}(x)=1$, $\forall x\in [0, b-a]$.

\subsection{The convolution max-metric}\label{section_convolution}
The max-metric can also be given in a convolution form. To construct this form of the metric we consider an arbitrary continuous, smooth, positive kernel $\mathcal{K} \colon \mathbb{R} \rightarrow \mathbb{R}^+$, with the properties that it is strictly increasing for $x<0$ and strictly decreasing for $x>0$, with $\mathcal{K}(0)=p$ finite and positive. We thus have $0 \leq \mathcal{K}(x) \leq p, \: \forall x \in \mathbb{R}$.
Typically, $\mathcal{K}$ is an exponential,
\begin{equation}
\mathcal{K}_E(x)=\exp\left(-\frac{|x|}{\tau}\right),\label{equation_K_exponential}
\end{equation}
with $\tau$ being a positive parameter. We convolve the two spike trains $T$ and $\bar{T}$ with the filtering kernel $\mathcal{K}$ to obtain
\begin{align}
f(x) &= \sum_{i=1}^{n}  \mathcal{K}(x-t^{(i)})\\
\bar{f}(x) &= \sum_{i=1}^{\bar{n}} \mathcal{K}(x-\bar{t}^{(i)}).
\end{align}
We denote by $\mathcal{F}_{[a, b]}$ the set of all possible filtered spike trains from $\mathcal{S}_{[a, b]}$.

We also consider a function $\mathcal{H} \in \mathbb{B}$ that is strictly positive, that is derivable on $(0,b-a)$ and that has a bounded derivative.

The convolution max-metric is defined as
\begin{equation}\label{m1}
d_{c}(T,\bar{T})= \int_a^b  \sup_{x \in [a,b]} \left\{ | f(x) - \bar{f}(x)| \; \mathcal{H}(|s-x|)\right\} \textrm{d}s.
\end{equation}
Figure \ref{fig:maxmetricconv} shows how the distance $d_c$ between two spike trains is computed. In Appendix \ref{section_analysis_convolution_max_metric} we show that $d_c$ is finite and that it satisfies the properties of a metric.

\section{Localized metrics}
In the case of the max metric, with or without convolution, the use of the kernel $\mathcal{H}$ served the purpose of providing a local perspective, around each point within $[a, b]$, of the distance between the spike trains. These local perspectives were then integrated in the final distance. In this section we introduce different metrics that also depend on a kernel $\mathcal{L}$, but for which the kernel has a different purpose. Here, the kernel may be regarded as a magnifying glass to be used to focus on one specific area of the spike trains. The kernel  should be a continuous, positive function, $\mathcal{L}: [0, b-a]\rightarrow \mathbb{R}^{+}$.
Similarly with $\mathcal{H}$, because $\mathcal{L}$ is a continuous function with bounded support, it is bounded, i.e.
\begin{equation}
0 \leq \mathcal{L}(x) < m < \infty, \; \forall \; x \in [0, b-a].\label{equation_l_conditions}
\end{equation}

Such a metric is biologically relevant if, for example, we take into consideration how a neuron responds to input spikes.  Recent spikes influence more the neuron than old ones. If we would like to measure the distance between two spike trains according to how the differences between them influence the activity of a neuron at a particular moment of time, recent differences should be taken into account with a greater weight than differences in the distant past. For the localized metrics, $\mathcal{L}$ could thus model the shape of postsynaptic potentials (PSP) that reflects the dynamics of the effect of one presynaptic spike on the studied neuron.
Thus, $\mathcal{L}$ could typically be an exponential, $\mathcal{L}_E=\mathcal{H}_E$ (Eq. \ref{equation_H_exponential}), an alpha function,
\begin{equation}
\mathcal{L}_\alpha(x)=\frac{x}{\tau^2} \; \exp\left(-\frac{x}{\tau}\right), \label{equation_H_local_alpha}
\end{equation}
a difference between two exponentials,
\begin{equation}
\mathcal{L}_D(x)=\frac{\tau}{\tau-\tau_s} \; \left[ \exp\left(-\frac{x}{\tau}\right) - \exp\left(-\frac{x}{\tau_s}\right) \right], \label{equation_H_local_double}
\end{equation}
or, if we model the postsynaptic potential generated in an integrate-and-fire neuron by a synaptic current that is a difference between two exponentials,
\begin{align}
\mathcal{L}_I(x)=\frac{\tau}{\tau_s-\tau_r} \;  & \left\{ \frac{\tau_s}{\tau-\tau_s} \; \left[ \exp\left(-\frac{x}{\tau}\right) - \exp\left(-\frac{x}{\tau_s}\right) \right] \right. \nonumber \\
 - & \;\; \left.\frac{\tau_r}{\tau-\tau_r} \; \left[ \exp\left(-\frac{x}{\tau}\right) - \exp\left(-\frac{x}{\tau_r}\right) \right] \right\},
 \label{equation_H_local_double_iaf}
\end{align}
where $\tau$, $\tau_s$, and $\tau_r$ are positive parameters.

\subsection{Localized max-metric}
We introduce the localized max-metric as
\begin{equation}\label{m4}
d_{l}(T,\bar{T}) = \int_{a}^{b} \mathcal{L}(b-s) \sup_{x \in [s,b]} | d(x,T) - d(x,\bar{T})|  \; \textrm{d}s.
\end{equation}
Figure \ref{fig:localmaxmetric} shows how the distance $d_{l}$ between two spike trains is computed. The differences between the spike trains that account the most for the distance are those that are close to $b$. The shape of $\mathcal{L}$ has a high impact on the behavior of the metric.

In Appendix \ref{section_analysis_localized_max_metric} we show that the distance $d_{l}$ is finite and that it satisfies the properties of a metric.

\begin{figure}[h!p]
\hfill
\begin{center}
\includegraphics[width=\textwidth]{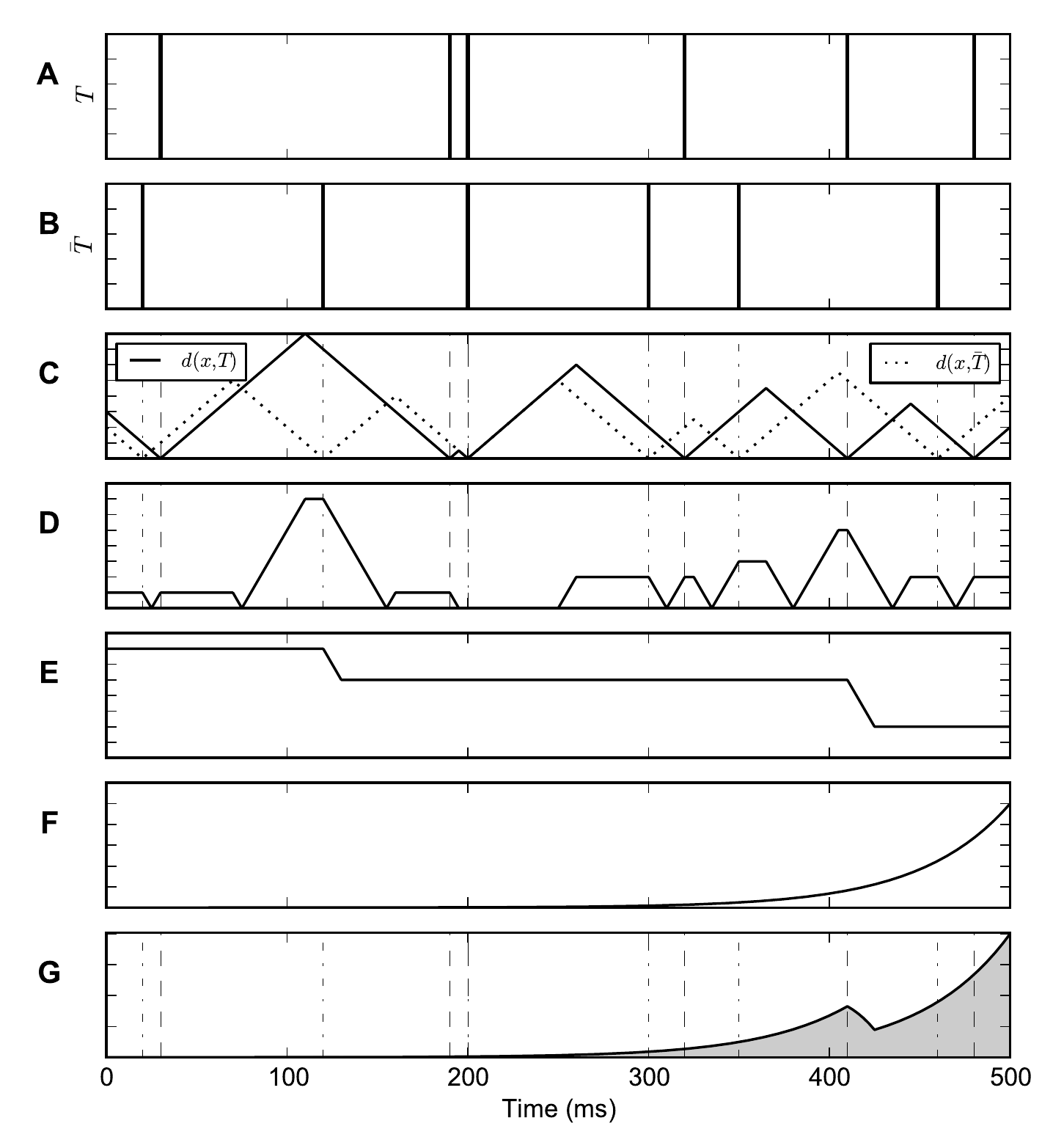}
\end{center}
\caption{The localized max-metric. (A) Spike train $T=\{30,190,200,320,410,480\}$ ms. Each spike time is represented as a vertical bar. (B) Spike train $\bar{T}=\{20,120,200,300,350,460\}$ ms. (C) The  distances between a timing $x$ and the spike trains, $d(x,T)$ and $d(x,\bar{T})$, as function of $x$. (D) The difference $ \left| d(x,T) - d(x,\bar{T}) \right|$ as a function of $x$.  (E) The supremum, $\sup_{x \in [s,b]} \left|d(x,T) - d(x,\bar{T}) \right|$, as a function of $s$. (F) The kernel $\mathcal{L}(b-s)$ as function of $s$, an exponential  (Eq. \ref{equation_H_exponential}) with a decay constant $\tau=50$ ms. (G) The supremum weighted by the kernel, $\mathcal{L}(b-s) \sup_{x \in [s,b]} | d(x,T) - d(x,\bar{T})|$, as a function of $s$. The localized max-metric distance $d_l(T, \; \bar{T})$ is the area under this curve. In (C), (D), and (G), the dashed vertical lines represent the timing of spikes in $T$ and $\bar{T}$.}
\label{fig:localmaxmetric}
\end{figure}

\subsection{Localized modulus-metric}
The modulus metric can also be given in a localized form:
\begin{equation}\label{m3}
d_n(T,\bar{T}) = \int_{a}^{b} |d(s,T) - d(s,\bar{T})| \; \mathcal{L}(b-s)  \; \textrm{d}s.
\end{equation}
In Appendix \ref{section_analysis_localized_modulus_metric} we show that $d_n$ is finite and that it satisfies the properties of a metric.

\subsection{Localizing the van Rossum metric}
A localization by a kernel $\mathcal{L}$ similar to the one we applied in Eqs. \ref{m4} and \ref{m3} can also be applied to existing metrics.
Let $T, \bar{T} \in \mathcal{S}_{[a,b]}$. Consider the \citet{R01} distance defined as
\begin{equation}
d_R(T,\bar{T}) = \int_{-\infty}^{\infty} (g(s)-\bar{g}(s))^2 \;  \textrm{d}s,
\end{equation}
where
\begin{align}
g(s) &= \sum_{i=1}^{n}  H(s-t^{(i)}) \; \mathcal{K}_E(s-t^{(i)})\\
\bar{g}(s) &= \sum_{i=1}^{\bar{n}} H(s-\bar{t}^{(i)}) \; \mathcal{K}_E(s-\bar{t}^{(i)}),
\end{align}
$H$ is the Heaviside step function, $H(x)=0$ if $x < 0$ and $H(x)=1$ if $x\geq0$, and $\mathcal{K}_E$ is defined in Eq. \ref{equation_K_exponential}.
When localized with $\mathcal{L}$ the distance becomes
\begin{equation}\label{mvr}
d_{Rl}(T,\bar{T}) = \int_{-\infty}^{b} (g(s)-\bar{g}(s))^2 \; \mathcal{L}(b-s) \; \textrm{d}s.
\end{equation}
Here $\mathcal{L}$ may be chosen to have the same qualitative properties as the kernels used in Eqs. \ref{equation_H_local_alpha}--\ref{equation_H_local_double_iaf}.

\section{Simulation results}\label{sec3}
We analyzed the behavior of the introduced metrics through computer simulations using simple setups. Across all simulations, $\mathcal{H}$ was an exponential (Eq. \ref{equation_H_exponential}) with $\tau = 10$ ms. For the localized metrics $d_l$ and $d_n$, $\mathcal{L}$ also was an exponential (Eq. \ref{equation_H_exponential}) with $\tau = 20$ ms. The convolution kernels for the $d_c$ and van Rossum distances were chosen as exponentials (Eq. \ref{equation_K_exponential}) with $\tau = 10$ ms. The width of the Gaussian filter for the Schreiber et al. distance was 10 ms. For the Victor \& Purpura distance we set $q=0.2$ ms$^{-1}$, except in Figs. \ref{fig:jitter_removal_1} -- \ref{fig:jitter_removal_3} where we also used $q=0.001$ ms$^{-1}$ as in \cite{dauwels2009}. We set $a=0$ ms and $b$ the maximum length of spike trains (either 200 or 500 ms, except in Fig. \ref{fig:speed} where the length was variable).

\subsection{Inserting or shifting one spike}
We computed the distances between a particular spike train $T$ and a spike train $\bar{T}$ obtained from it by either inserting or shifting one spike. In the insertion case, $\bar{T}$ was generated by inserting a spike into $T$ at various timings. In the shifting case, $\bar{T}$ was generated by shifting a particular spike of $T$. The distance was plotted against the time of the inserted spike or of the shifted spike to see how the change is reflected by the metrics. To compute the distance we used the introduced metrics, a simple spike count distance ($c$), the \citet{R01} ($d_R$), \citet{VP96,VP97} ($d_{VP}$), \citet{schreiber2003} ($s$), and Pompeiu-Hausdorff ($h$) distances, as well as the ISI-distance ($k_i$) and the improved SPIKE-distance ($k_s$) by \citet{kreuz2007,kreuz2013}. The spike count distance is defined as
\begin{equation}
c(T,\bar{T}) = \frac{|n-\bar{n}|}{\max ( n,\bar{n})},
\end{equation}
where $n$ and $\bar{n}$ are the number of spikes in each train. The spike trains were 200 ms long.

The results for the insertion case are presented in Figure \ref{fig:a1}. The Victor \& Purpura distance was constant since the cost of adding and removing a spike is fixed at 1 regardless of its timing. Similarly, the van Rossum metric was insensitive to the time of the inserted spike, a result which can be also shown analytically \cite{R01}. The spike count distance remained constant regardless of where the spike was inserted. The results were qualitatively different in the case of our newly introduced distances, with the exception of the convolution max-metric, and in the case of the Kreuz et al. and Schreiber et al. metrics. In the case of the Pompeiu-Hausdorff distance, max-metric, modulus-metric, of the localized variants of the max-metric and the modulus-metric, and of the Schreiber et al. and Kreuz et al. metrics, the insertion time of the spike had a significant impact on the outcome (Figure \ref{fig:a1}). When the inserted spike overlapped an existing spike, the Schreiber et al. distance had a low value but remained non-zero, while our new metrics, with the exception of the convolution max-metric, as well as the Kreuz et al. metrics, returned a zero distance. It can also be seen that the localized distances were strongly influenced by the shape of the $\mathcal{L}$ kernel.

The results for the shifting case are presented in Figure \ref{fig:a2}. When the spike at $t^{(4)}$ was shifted, the Victor \& Purpura and van Rossum distances were dependent only on the width of the shifting interval. These results are confirmed by analytical derivations \cite{VP96, R01}. As in the previous case, the spike count distance was insensitive to the shift operation and remained zero since the number of spikes did not change. In contrast,  our newly introduced distances, with the exception of the convolution max-metric, and the Schreiber et al. and Kreuz et al. metrics showed a dependence not only on the width of the shifting interval but also on the particular timing of the shifted spike. The results are similar to the ones in Fig. 4 of \cite{kreuz2011}.

\begin{figure}[h!p]
\hfill
\begin{center}
\includegraphics[width=0.8 \textwidth]{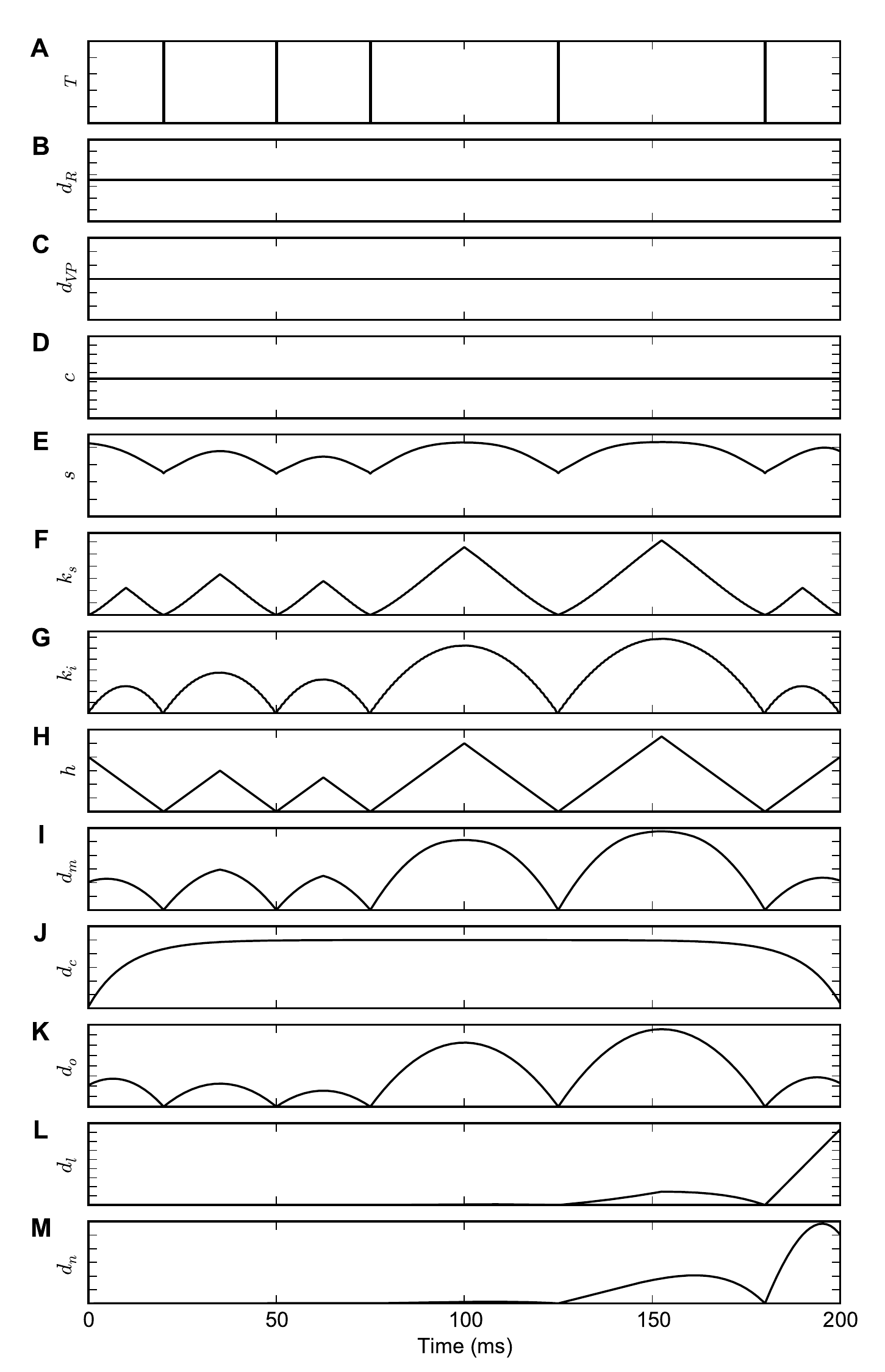}
\end{center}
\caption{Metric comparison - inserting a spike. We computed the distance between the spike train $T= \{ 20, 50, 75, 125, 180\}$ ms and one obtained from this spike train by inserting a spike at different locations. At each time $x \in [0, 200]$ ms a spike was inserted to generate $\bar{T}$ and the distance between $T$ and $\bar{T}$ was computed and plotted against $x$. (A) The spike train $T$. (B) The van Rossum distance. (C) The Victor \& Purpura distance. (D) The spike count distance. (E) The Schreiber et al. distance. (F) The Kreuz et al. improved SPIKE-distance. (G) The Kreuz et al. ISI-distance. (H) The Pompeiu-Hausdorff distance. (I) The max-metric. (J) The convolution max-metric. (K) The modulus-metric. (L) The localized max-metric. (M) The localized modulus-metric.}
\label{fig:a1}
\end{figure}

\begin{figure}[h!p]
\hfill
\begin{center}
\includegraphics[width=0.8 \textwidth]{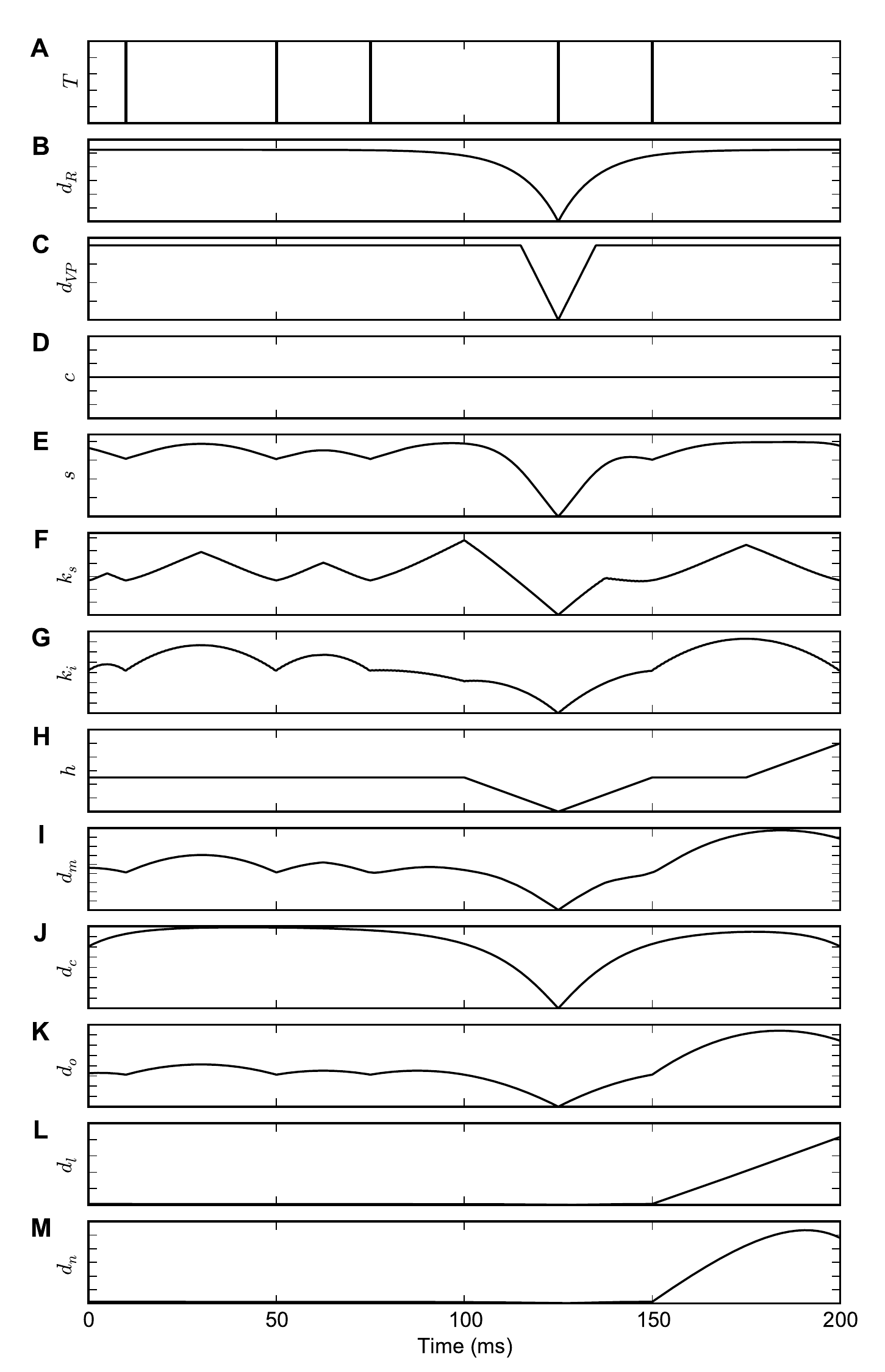}
\end{center}
\caption{Metric comparison - shifting a spike. We computed the distance between the spike train $T= \{ 10, 50, 75, 125, 150\}$ ms and one obtained from this spike train by shifting the spike at $t^{(4)}=125$ ms. The spike was shifted at timings $x \in [0, 200]$ ms to generate $\bar{T}$ and the distance between $T$ and $\bar{T}$ was computed and plotted against $x$. (A) The spike train $T$. (B) The van Rossum distance. (C) The Victor \& Purpura distance. (D) The spike count distance. (E) The Schreiber et al. distance. (F) The Kreuz et al. improved SPIKE-distance. (G) The Kreuz et al. ISI-distance. (H) The Pompeiu-Hausdorff distance. (I) The max-metric. (J) The convolution max-metric. (K) The modulus-metric. (L) The localized max-metric. (M) The localized modulus-metric.}
\label{fig:a2}
\end{figure}

\subsection{Bursts}
We generated a template spike train $T$
containing three bursts and one isolated spike. We computed, in 6 setups, using various metrics, the distance between $T$ and another spike train $\bar{T}$ obtained from $T$ by either: removing one spike from a burst;  inserting one spike into a burst; or removing or inserting one spike not belonging to the bursts. After computing the distances for each of the setups, the distances for each metric were normalized to the maximum distance for that particular metric among the setups. The normalized distances are shown, for each setup, in Figure \ref{fig:bursts}. When a spike is added to a burst or removed from a burst, the max-metric and modulus-metric distances, as well as the Kreuz et al. distances, are close to zero. Those distances become non-negligible when a solitary spike, far from a burst, is removed, or a new spike is added far from a burst. The normalized convolution max-metric and the van Rossum distances are close to 1, and the Victor \& Purpura distance is exactly 1, in all setups. The Schreiber et al. distance exhibited an intermediate behavior, but it also remained non-negligible in all setups.

\begin{figure}[h!p]
\hfill
\begin{center}
\includegraphics[width=\textwidth]{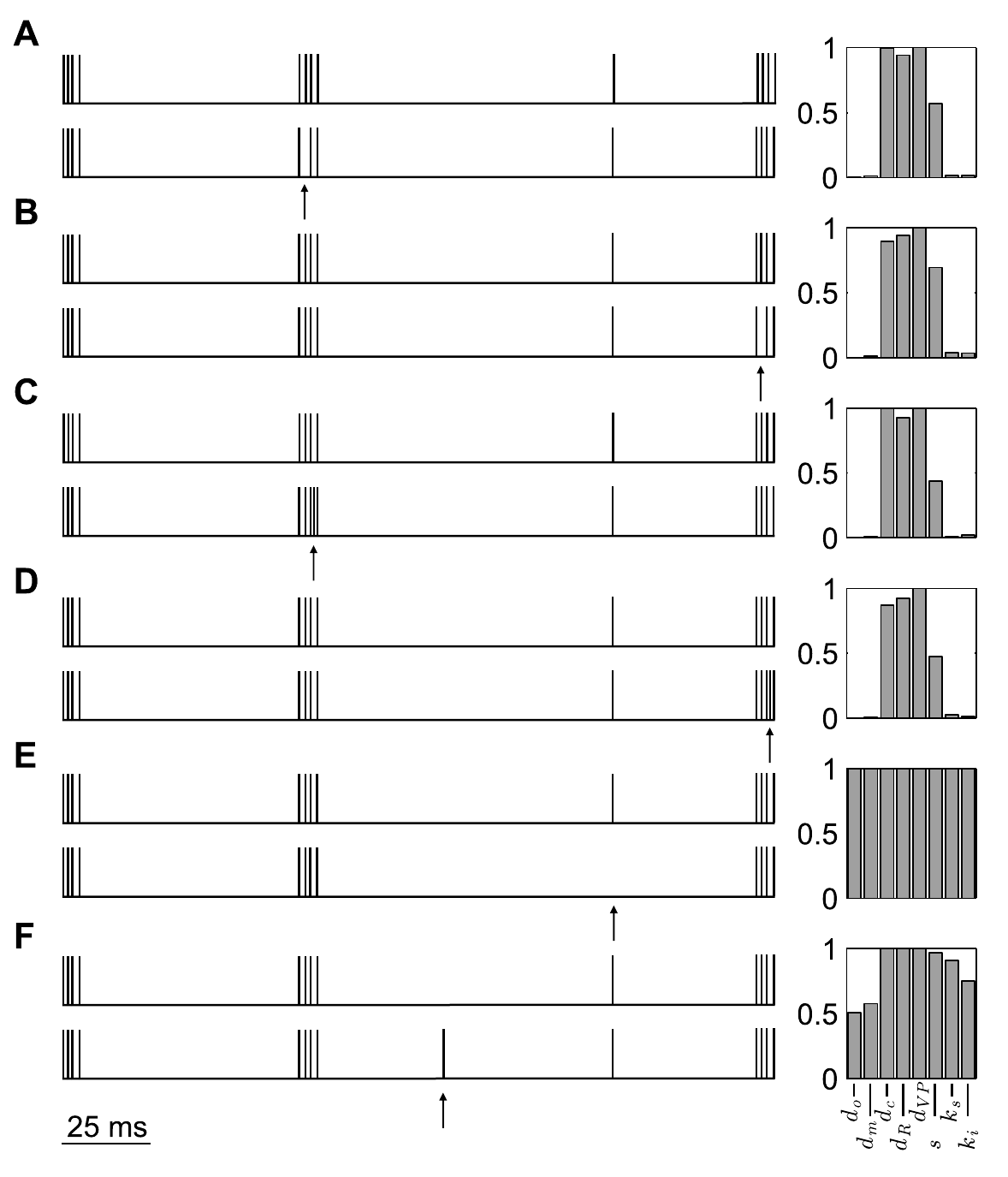}
\end{center}
\caption{Dealing with bursts. We measured distances between a template spike train that includes bursts and a spike train obtained from the template by adding or removing one spike. Arrows indicate where the spike was added or removed. Left: the pair of spike trains. Right: the normalized distances ($d_o$ - modulus metric; $d_m$ - max metric; $d_c$ - convolution max-metric; $d_R$ - van Rossum; $d_{VP}$ - Victor \& Purpura; $s$ - Schreiber et al.; $k_s$ - Kreuz et al. improved SPIKE-distance; $k_i$ - Kreuz et al. ISI-distance). (A), (B) One spike removed from a burst. (C), (D) One spike inserted into a burst. (E) One isolated spike removed. (F) One isolated spike inserted.}
\label{fig:bursts}
\end{figure}

\subsection{Discriminating timing precision vs. event reliability}
Timing precision and event reliability are distinct characteristics of the variation of a spike train \cite{TFS08}. In computational neuroscience there is a need for spike train metrics to both characterize the overall variability as well as for parsing out the precision and reliability separately \cite{toups2012}. We investigated how our new metrics and existing metrics perform in discriminating precision vs. event reliability \cite{dauwels2009}. We generated a 200 ms template spike train $T$ using a Poisson process with a rate of 100 Hz. We also generated various spike trains $\bar{T}$ obtained from $T$ by both: (a) applying to spikes from $T$ a Gaussian jitter with zero mean and variance $\sigma$ in the range of 0 and $\sigma_m=20$ ms and (b) removing spikes with a probability $p$ ranging from 0 to $p_m=80 \%$. For each $(\sigma, p)$ pair we ran 100 trials, where in each trial we generated randomly a new $\bar{T}$. For each trial, we computed the distances between $T$ and $\bar{T}$ using various metrics and we then averaged the results over the trials. For each metric, average distances were then normalized to the maximum average distance across the $(\sigma, p)$ parameters. The results are presented in Figs. \ref{fig:jitter_removal_1}--\ref{fig:jitter_removal_4}. Fig. \ref{fig:jitter_removal_1} represents the normalized average distances as a function of $\sigma$ and $p$. Figs. \ref{fig:jitter_removal_2} and \ref{fig:jitter_removal_3} represent sections trough the graphs in Fig. \ref{fig:jitter_removal_1}, for illustrating more clearly the dependencies of distances on reliability and precision.

If Fig. \ref{fig:jitter_removal_4}, we tried to illustrate how reliability compares to the precision in determining the variability of the distances. For a metric $\rho$, it is not possible to compare directly the partial derivatives $\partial \rho/\partial p$ and $\partial \rho/\partial \sigma$ because they are variables of different physical dimensions: $p$ is dimensionless, while $\sigma$ represents a time interval. In order to compare them, we take into consideration that the relevant intervals on which $p$ and $\sigma$ vary are practically bounded. Choosing $\sigma_m$ and $p_m$ is theoretically arbitrary, but in practice $p_m$ is something slightly less than 1 and the relevant $\sigma_m$ is constrained by the time scales of the considered spike trains and of the distances that depend on time-scale-like parameters. We use these practical bounds in order to commensurate the variance with respect to the two parameters.  By dividing the intervals $[0, p_m]$ and $[0, \sigma_m]$ into the same number of bins $N$ we get two-dimensional pixels of size $\delta p=p_m/N$ and $\delta \sigma=\sigma_m/N$ on which we can consider that the comparison of the variation of $d$ along the two axes can be compared. We computed
\begin{equation}
\delta \rho (p, \sigma) = \frac{\partial \rho(p, \sigma) }{\partial p} \; \delta p - \frac{\partial \rho(p, \sigma) }{\partial \sigma} \; \delta \sigma,
\end{equation}
where $\rho$ was replaced with the considered metrics. In areas of the $p$, $\sigma$ space where $\delta \rho$ is positive, we may say that the metric is more sensitive to reliability, while in areas where it is negative the metric is more sensitive to the precision of spikes. The interpretation of results should take into consideration the caveat that changes of the $\sigma_m/p_m$ ratio may change the sign of $\delta \rho$.

As expected, the Victor \& Purpura distance with $q=0.001$ ms$^{-1}$ does not depend at all on the precision, but just on reliability, since it basically counts the difference in the number of spikes between the spike trains. Our modulus-metric and max-metrics have a stronger dependence on reliability than on precision, on most of the considered range, except for high reliability (very low $p$) where the dependence on precision is still dominant. For the van Rossum and Schreiber et al. distances, the dependence on reliability and precision is somehow balanced. The Kreuz et al. distances and the Victor \& Purpura distance with $q=0.2$ ms$^{-1}$ have a stronger dependence on precision than on reliability. For all distances, except for the Victor \& Purpura distance with $q=0.001$ ms$^{-1}$, the sensitivity to reliability increases with the unreliability (with the probability $p$ of spikes not being fired). These results may be different for a different choice of the $\sigma_m/p_m$ ratio, of the time scales of the considered spike trains and of the parameters of the parameter-dependent distances. Our results are similar to the ones of \citet{dauwels2009} for the Victor \& Purpura distance $q=0.001$ ms$^{-1}$, but different for the Schreiber et al. distance, probably because \citet{dauwels2009} used a different approach for modeling unreliability of spikes.

\begin{figure}
	\centering
	\includegraphics[width=0.9 \textwidth]{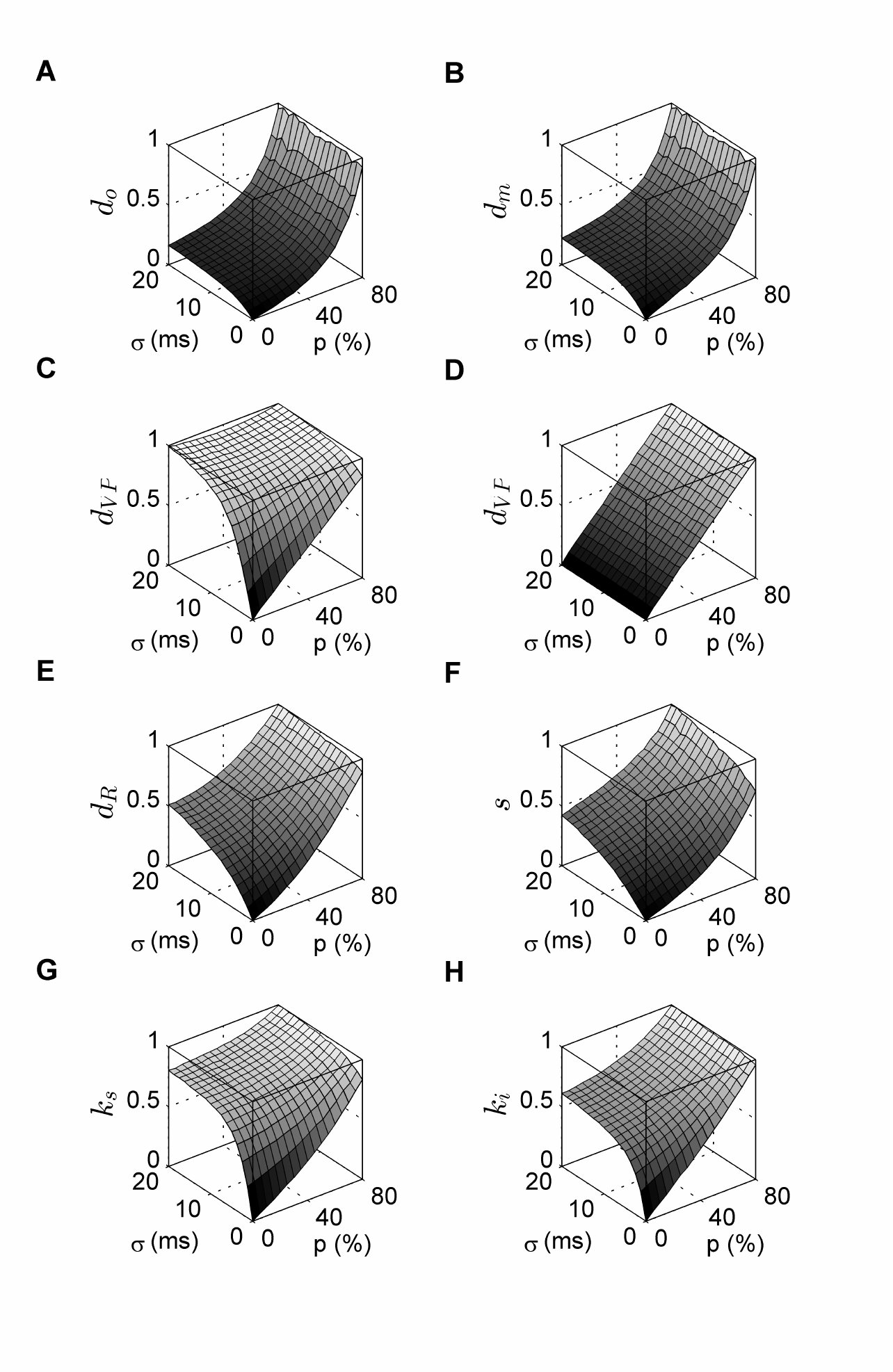}
	\caption{Normalized average distances between a template spike train and another one generated from it, as a function of the probability $p$ of removing spikes from the template spike train and of jitter $\sigma$ added to spikes. Distances are normalized to their maximum over the considered range. (A) Modulus-metric. (B) Max-metric. (C) Victor \& Purpura, $q=0.2$ ms$^{-1}$. (D) Victor \& Purpura, $q=0.001$ ms$^{-1}$. (E) Van Rossum. (F) Schreiber et al. (G) Kreuz et al. improved SPIKE-distance. (H) Kreuz et al. ISI-distance.}
\label{fig:jitter_removal_1}
\end{figure}
\begin{figure}
	\centering
	\includegraphics[width=0.9 \textwidth]{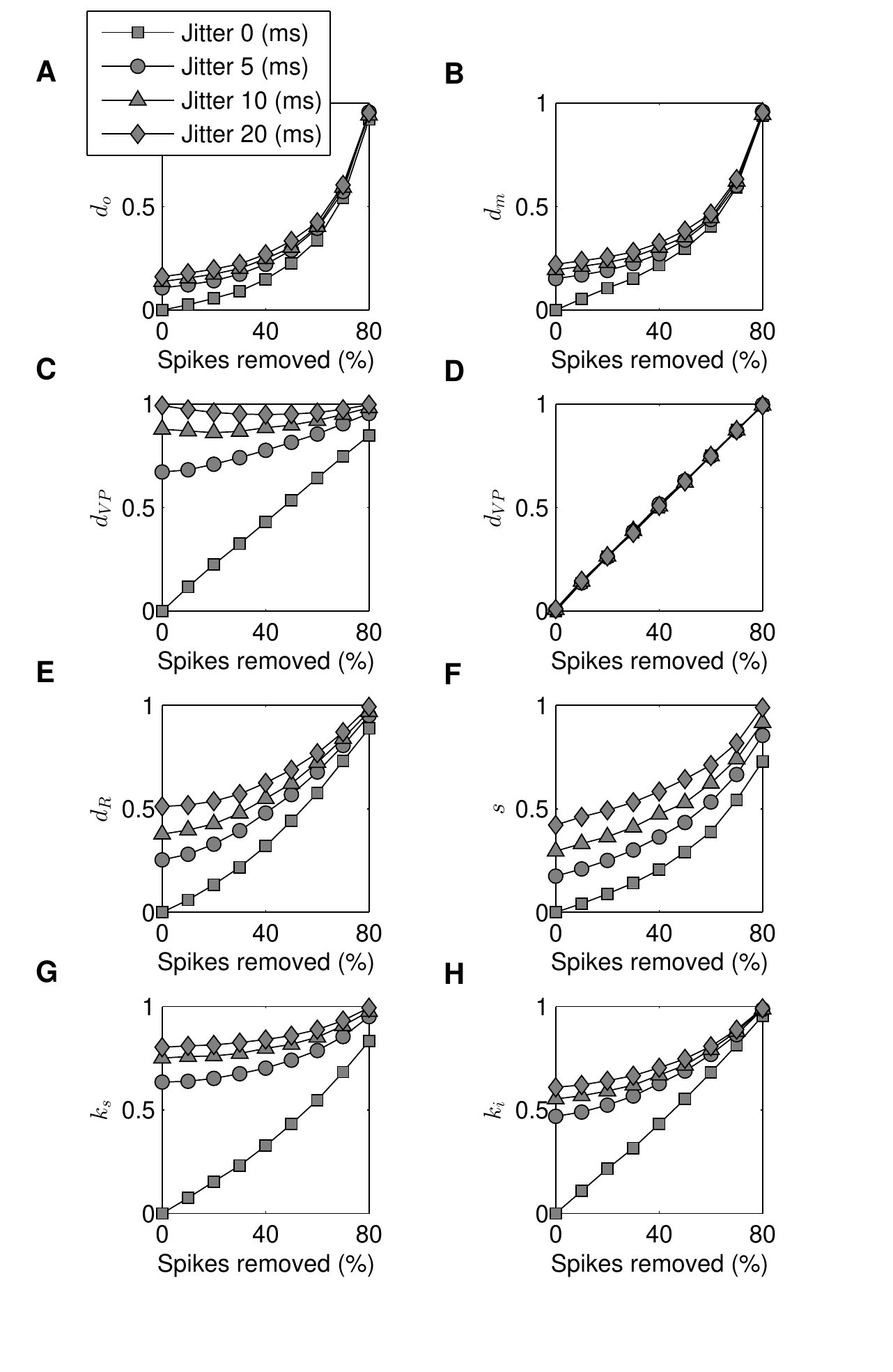}
	\caption{Normalized average distances as a function of the probability of removing spikes $p$, for several levels of jitter $\sigma$ (sections trough the panels in Fig. \ref{fig:jitter_removal_1}). (A) Modulus-metric. (B) Max-metric. (C) Victor \& Purpura, $q=0.2$ ms$^{-1}$. (D) Victor \& Purpura, $q=0.001$ ms$^{-1}$. (E) Van Rossum.  (F) Schreiber et al. (G) Kreuz et al. improved SPIKE-distance. (H) Kreuz et al. ISI-distance.}
\label{fig:jitter_removal_2}
\end{figure}
\begin{figure}
	\centering
	\includegraphics[width=0.9 \textwidth]{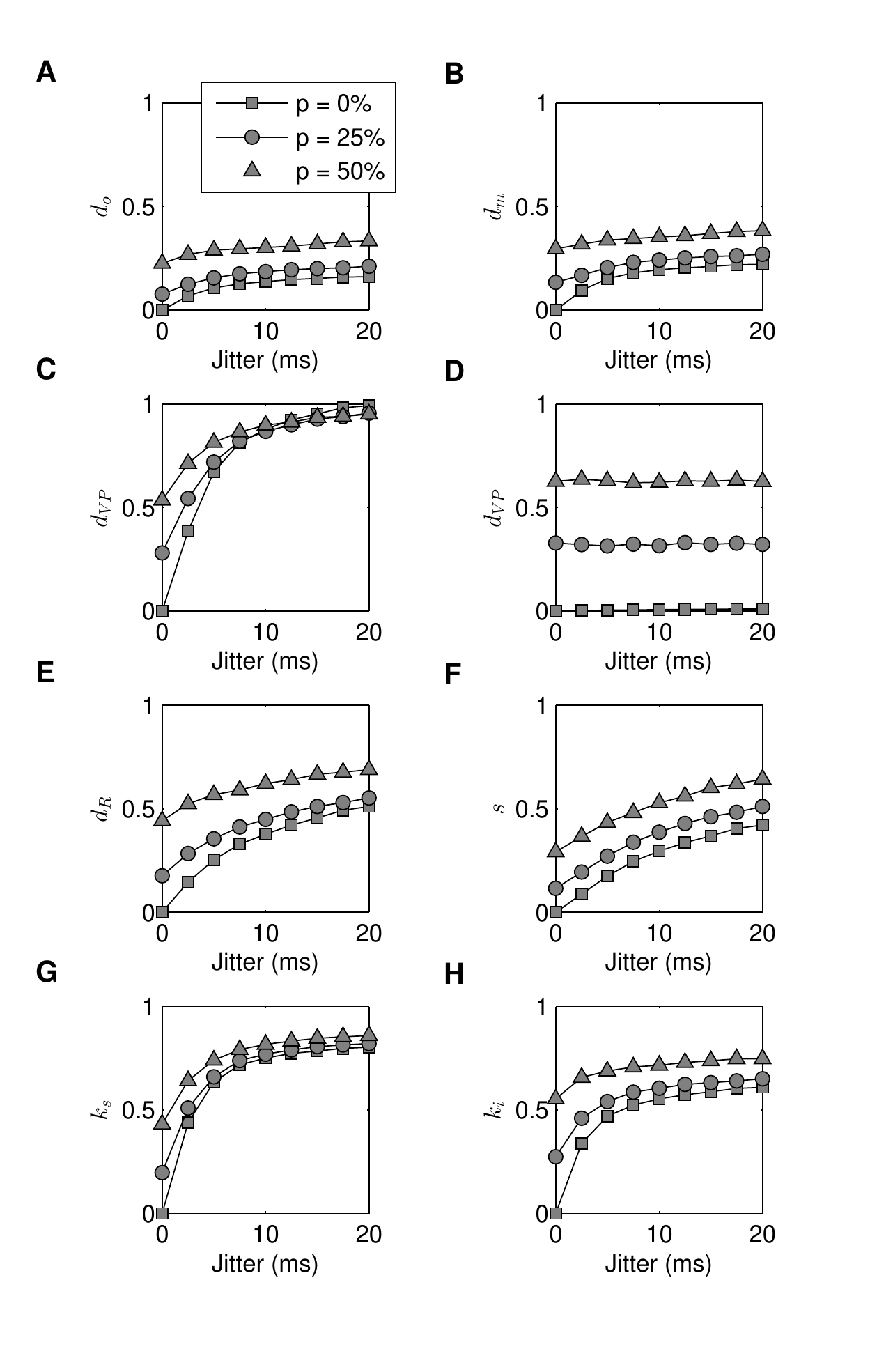}
	\caption{Normalized average distances as a function of jitter $\sigma$, for several levels of the probability of removing spikes $p$  (sections trough the panels in Fig. \ref{fig:jitter_removal_1}). (A) Modulus-metric. (B) Max-metric. (C) Victor \& Purpura, $q=0.2$ ms$^{-1}$. (D) Victor \& Purpura, $q=0.001$ ms$^{-1}$. (E) Van Rossum.  (F) Schreiber et al. (G) Kreuz et al. improved SPIKE-distance. (H) Kreuz et al. ISI-distance.}
\label{fig:jitter_removal_3}
\end{figure}
\begin{figure}
	\centering
	\includegraphics[width=0.9 \textwidth]{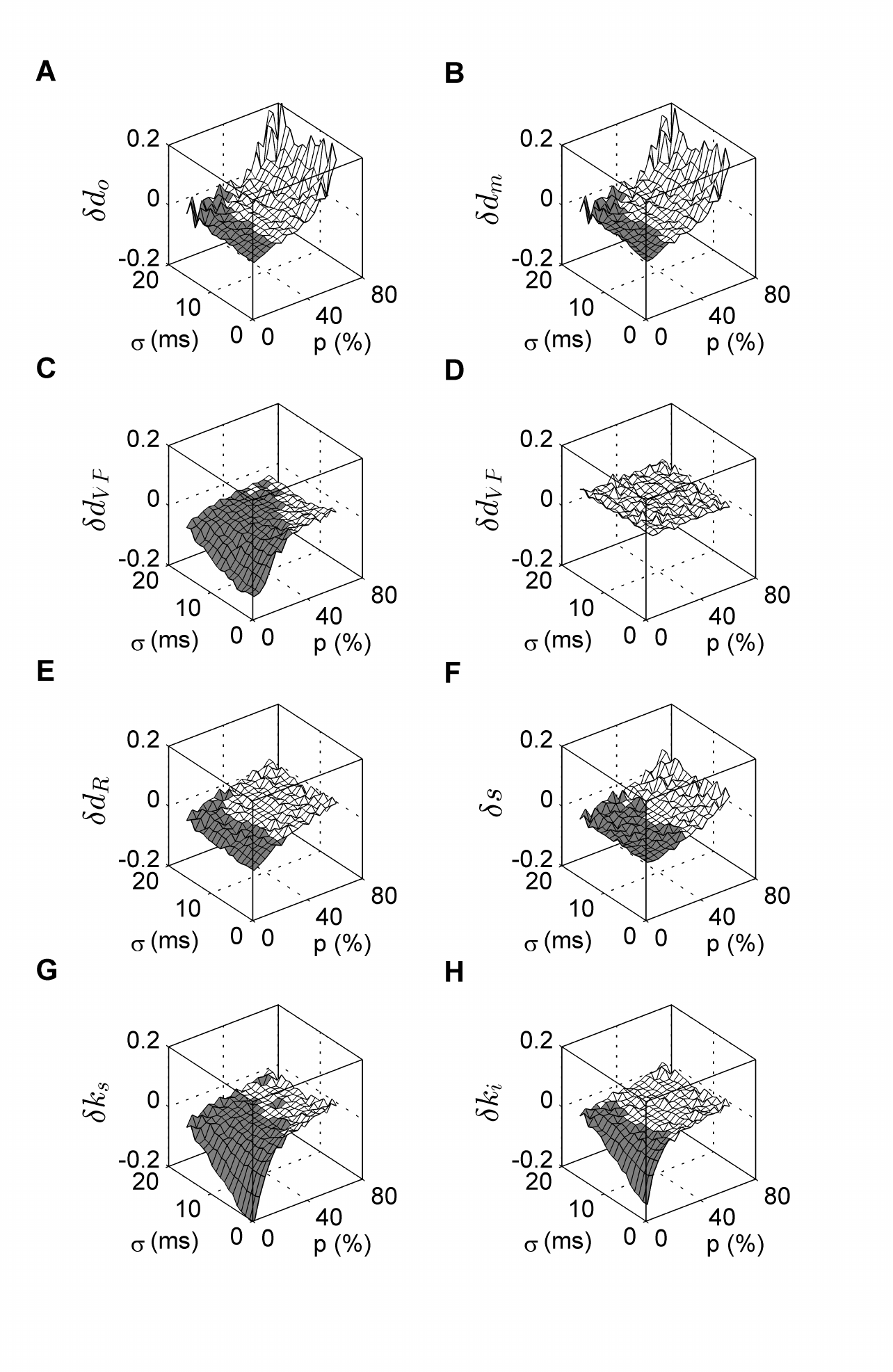}
	\caption{Illustration of the dominance of reliability over precision for determining the variability of distances (see main text). In white areas, the distance varies more with respect to reliability than with respect to precision. In dark areas, the dominance is reversed. (A) Modulus-metric. (B) Max-metric. (C) Victor \& Purpura, $q=0.2$ ms$^{-1}$. (D) Victor \& Purpura, $q=0.001$ ms$^{-1}$. (E) Van Rossum.  (F) Schreiber et al. (G) Kreuz et al. improved SPIKE-distance. (H) Kreuz et al. ISI-distance.}
\label{fig:jitter_removal_4}
\end{figure}

\subsection{Correlations}
We explored the correlation between the newly introduced metrics and the classical Victor \& Purpura and van Rossum distances. We generated a 500 ms Poisson spike train with a firing rate of 20 Hz. From this spike train, we generated a new one by adding a Gaussian jitter with zero mean and 20 ms variance. We considered only generated and jittered spike trains that contained 10 spikes. We then measured the distance between the original and the jittered spike train using various metrics. We repeated this in 1,000 trials, where for each trial the original spike train and the jitter were generated randomly. For each metric, distances were normalized to the mean value across samples. The results are displayed in Figure \ref{fig:s1}. Table \ref{tab:vr_cor_coef} shows the correlation coefficients between the max-metric, the modulus-metric, the convolution max-metric, the Schreiber et al. distance, the Kreuz et al. distances and, respectively, the van Rossum and Victor \& Purpura distances.

\begin{figure}[h!p]
\hfill
\begin{center}
\includegraphics[width=\textwidth]{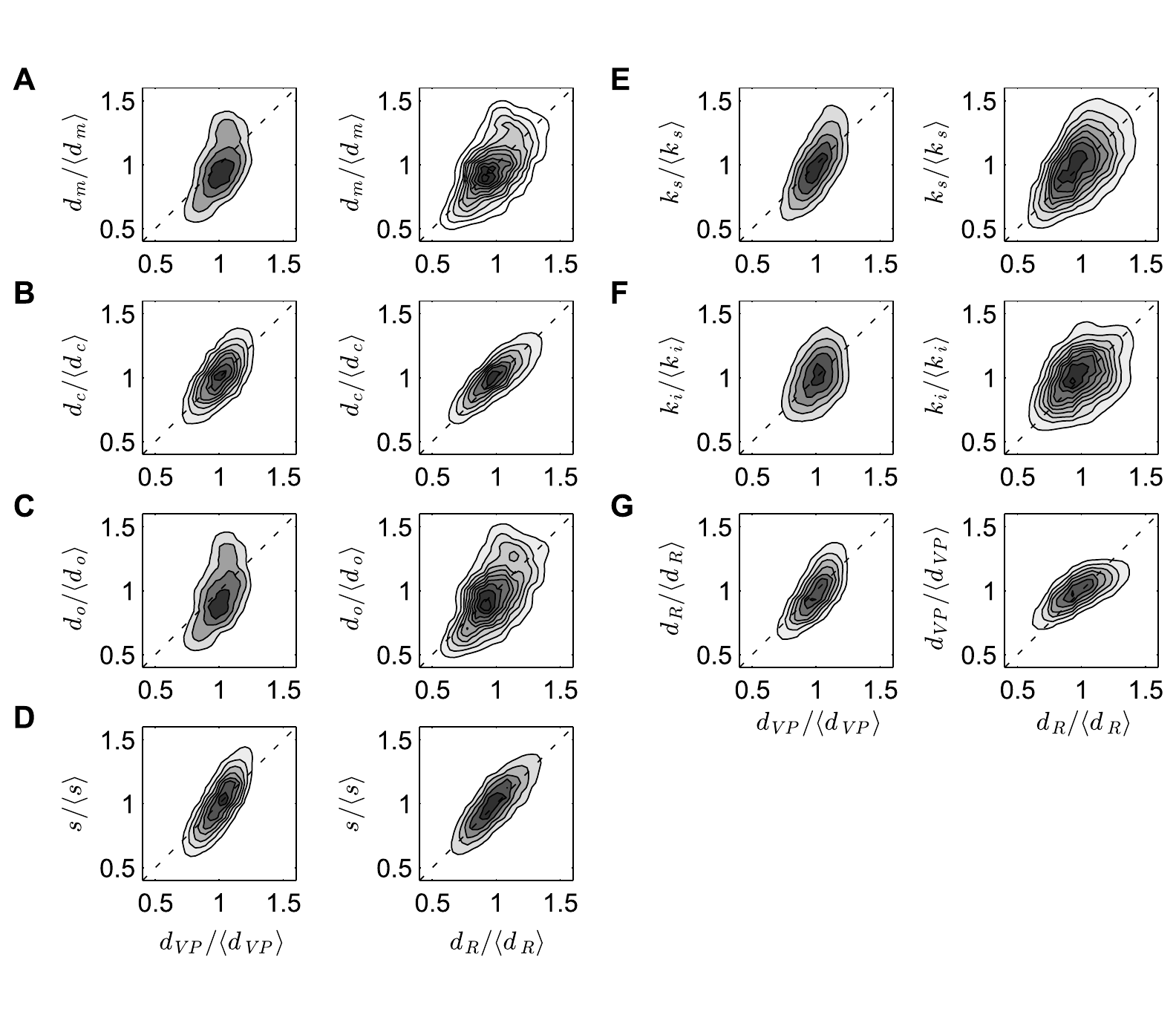}
\end{center}
\caption{Correlations between the Victor \& Purpura (left column) and the van Rossum (right column) metrics and, respectively, the distances computed with our newly introduced metrics, the Schreiber et al. and the Kreuz et al. metrics. We computed distances between Poisson spike trains and a jittered version of them. Figures represent the distribution of pairs of the normalized values of distances measured with the different metrics. (A) Max-metric. (B) Convolution max-metric. (C) Modulus-metric. (D) Schreiber et al. (E) Kreuz et al. improved SPIKE-distance. (F) Kreuz et al. ISI-distance. (G) Victor \& Purpura / van Rossum.}
\label{fig:s1}
\end{figure}

\begin{table}[h]
\caption{Correlation coefficients between the van Rossum and Victor \& Purpura distances and, respectively, other distances, computed from data presented in Figure \ref{fig:s1}.}
\begin{tabular}{|l|r|r|}
\hline
 Distance & $d_R$ correlation coefficient & $d_{VP}$ correlation coefficient\\
\hline
  $d_m$ & 0.54 & 0.48\\
\hline
  $d_o$ & 0.54 & 0.47\\
\hline
  $d_c$ & 0.84 &  0.81\\
\hline
  $s$ & 0.87 &  0.88\\
\hline
  $k_s$ & 0.51 &  0.70\\
\hline
  $k_i$ & 0.43 &  0.44\\
\hline
  $d_R$ & 1.00 &  0.78\\
\hline
  $d_{VP}$ &  0.78 & 1.00 \\
\hline
\end{tabular}
\label{tab:vr_cor_coef}
\end{table}

\subsection{Computation speed}
We computed the distances between pairs of randomly generated spike trains, the spike trains within a pair having the same number of spikes $n$. We varied $n$ from 5 to 500 while keeping constant the firing rate of the spike trains. The trains were generated by randomly choosing the $n$ firing times from a uniform distribution between 0 and $n\;T$, where T=35 ms. We measured the average time needed to compute distances, as a function of the number of spikes. Performance was measured using C++ implementations of the metrics, running on a Intel Core 2 processor. The Victor \& Purpura metric was computed using the algorithm in \cite{VP96}. The van Rossum metric was implemented using the exact, optimized algorithm presented in \cite{HK12} (A1) and a discrete-time integration with a timestep of 1 ms (A2), which turned out to be slightly faster than the optimized one. The Schreiber et al. and the Kreuz et al. metrics have also been computed using a 1 ms integration timestep. The modulus-metric was implemented using Algorithms 1 and 2. The code used for all metrics is freely available at \url{https://github.com/modulus-metric/}. The results are presented in Figure \ref{fig:speed}. In panel (A), the results were averaged over 1,000 trials, while in the other panels the results were averaged over 10,000 trials.

The simulations showed that the max-metric and the Schreiber et al. metric are relatively slow to compute. Those two metrics and the Victor \& Purpura metric require a computation time that grows more than linearly with the number of spikes. The other metrics have an approximately linear dependence on the number of spikes; we fitted them with a line and computed the proportionality coefficients in Table \ref{tab:speed}. The fastest distances or algorithms were, in order: modulus-metric (Algorithm 2); van Rossum (A2); Kreuz et al. ISI-distance; van Rossum (A1); modulus-metric (Algorithm 1); Kreuz et al. improved SPIKE-distance. It should be noted that, while the modulus-metric algorithms and the van Rossum A1 algorithm compute the distances exactly (within machine numerical precision), the other linear-time algorithms compute numerical approximations of the distances through discrete-time integration, with a precision that depends on the integration timestep.

 \begin{figure}
	\centering
	\includegraphics[width=\textwidth]{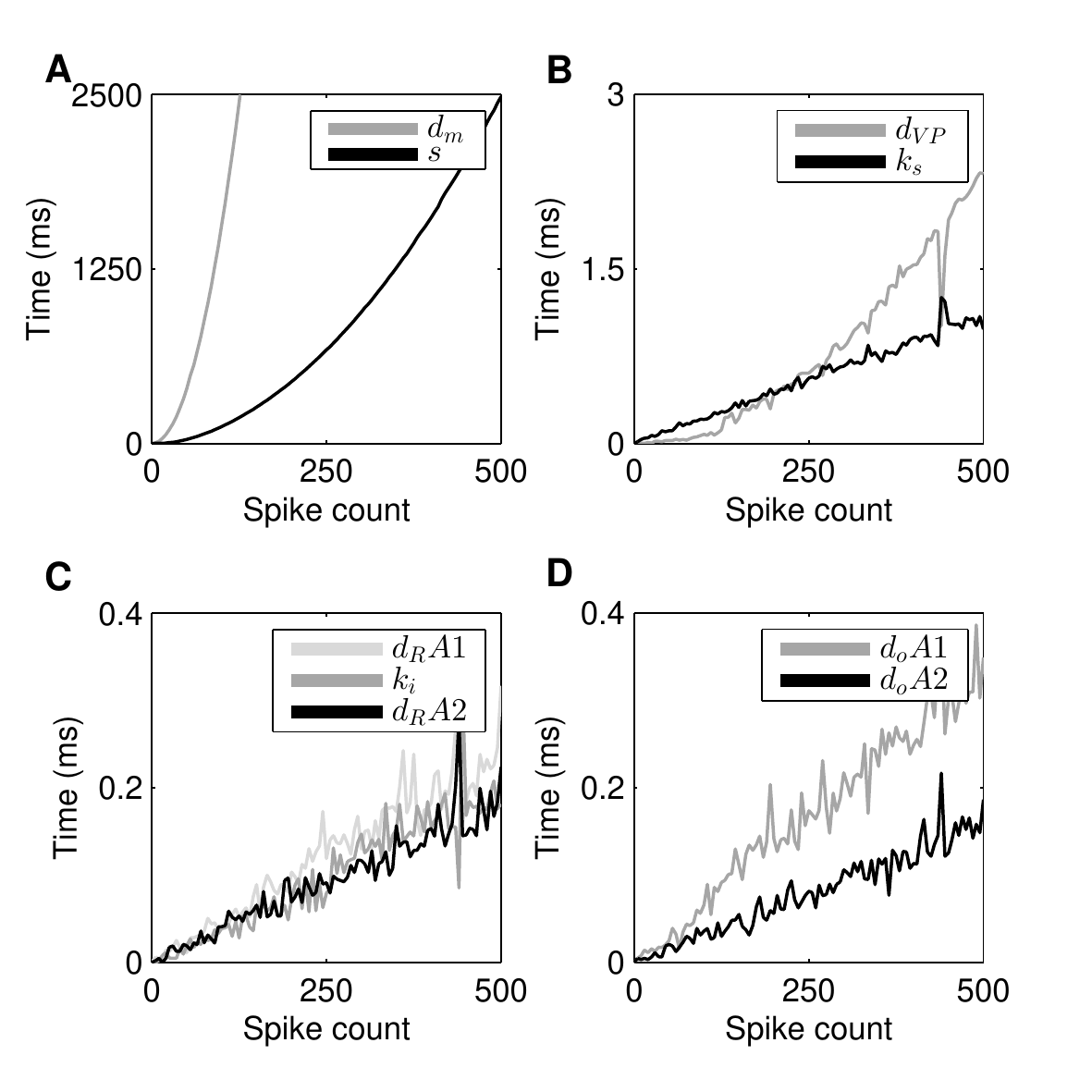}
	\caption{Time needed for computing various distances, as a function of the number of spikes in the spike trains. Note the different time scales. (A) $d_m$: max-metric; $s$: Schreiber et al. (B)  $d_{VP}$: Victor \& Purpura; $k_s$: Kreuz et al. improved SPIKE-distance. (C) $d_R$: van Rossum; A1: algorithm by \citep{HK12}; A2: discrete-time integration; $k_i$ - Kreuz et al. ISI-distance. (D) $d_o$: modulus metric; A1: Algorithm 1; A2: Algorithm 2.}
\label{fig:speed}
\end{figure}

\begin{table}[h]
\caption{Average computing time per spike in a spike train, for various distances. Notation and data as in Figure \ref{fig:speed}.}
\begin{tabular}{|l|r|}
\hline
 Distance (and algorithm) & Computing time per spike (ms)\\
\hline
  $d_o$ A2& 0.001570\\
\hline
  $d_R$ A2 & 0.001831\\
\hline
  $k_i$ & 0.001965\\
\hline
  $d_R$ A1 & 0.002437\\
\hline
  $d_o$ A1 & 0.003323\\
\hline
  $k_s$ & 0.010910\\
\hline
\end{tabular}
\label{tab:speed}
\end{table}

\section{Discussion}

The max-metric and the modulus-metric behave in a qualitatively different way than the classical van Rossum and Victor \& Purpura distances (Figures \ref{fig:a1}-\ref{fig:bursts}), but similar to the Kreuz et al. distances. Within a set of spike trains that are considered by the van Rossum and Victor \& Purpura metrics to be at equal distance from a reference spike train, the max-metric, the modulus-metric and the Kreuz et al. metrics can distinguish a range of distances that reflect similarities in the structure of the compared spike trains (Figures \ref{fig:a1}, \ref{fig:a2}). When comparing spike trains that include bursts, the max-metric, the modulus-metric and the Kreuz et al. metrics ignore differences in the number and position of spikes inside the bursts, while this kind of differences are considered by the van Rossum and Victor \& Purpura metrics as significant as in the case that differences exist outside the bursts (Figure \ref{fig:bursts}). This makes the max-metric, the modulus-metric  and the Kreuz et al. metrics particularly suitable for measuring distances in spike trains where information is encoded in the timing of bursts or solitary spikes, but not in the internal structure of bursts. This is the case in some experiments, such as \cite{reinagel1998}, where bursts regarded as unitary events encoded more information per event than otherwise, or such as \cite{kepecs2003}, where the timing of the first burst spike carried 70\% of the information and the spike count only 22\% of the information. However, in other cases the internal structure of bursts does carry information \cite{krahe2004}. The Schreiber et al. metric has a behaviour that is intermediate between the one of the max-metric, the modulus-metric and the Kreuz et al. metrics, on one side, and the van Rossum and Victor \& Purpura metrics, on another side.

While the max-metric depends on a kernel $\mathcal{H}$ which can be particularized to cause distinct behaviors, and the van Rossum, Victor \& Purpura and Schreiber et al. distances also depend on parameters that must be chosen by their users, the modulus-metric does not depend on any parameters, similarly to the \citet{kreuz2009,kreuz2011,kreuz2013} metrics. The lack of parameters allows one to start analyzing data immediately, without the need of preprocessing it in order to find the appropriate parameters. A parameter-free distance also gives a more objective measure, that does not depend on any assumptions to be made by the experimenter. In some cases, when the time scales vary during an experiment, no single time scale may characterize the spike trains, and thus a time-scale-independent measure may be preferable.

We have also shown that the modulus-metric can be computed faster than any of the other considered metrics, through an algorithm that operates in a time that depends linearly on the number of spikes in the considered spike trains. This fast algorithm computes the distance exactly (within machine precision), not as a nu\-me\-ri\-cal\-ly-approximated discrete-time integration, as some algorithms for other metrics do.

The convolution-metric that we introduced, although analytically similar to the max-metric, is qualitatively similar to the van Rossum distance. A qualitative difference between the convolution-metric and the van Rossum distance appears when the differences between the spike trains are localized near the ends of the integration interval, and this is a simple consequence of the difference between the bounded integration interval for the convolution-metric and the infinite integration interval for the van Rossum distance.

We have considered only spike trains having non-overlapping spikes. If we relax this constraint, for our newly introduced distances, with the exception of the convolution max-metric, we get a zero distance between a spike train and a second one generated from the first by adding an extra spike to the first, overlapping an existing spike (Figure \ref{fig:a1}). This is due to the distance $d$ between an arbitrary timing and a spike train in Eq. \ref{equation_distance_d} which does not distinguish between overlapping spikes in a train. Thus, if we relax the constraint of not allowing non-overlapping spikes, these distances become pseudo-metrics because there might be a zero distance between two spike trains that differ through overlapping spikes. However, the case of overlapping spikes is biologically implausible if we consider spike trains fired by single neurons. If it is enough that the distances are pseudo-metrics, we may also relax some of the conditions of the kernels, such as the requirement for $\mathcal{L}$ to be strictly positive on $(0, b-a]$ or the conditions on $\mathcal{K}$.

For our metrics, when the integrating interval $[a, \; b]$ extends beyond the interval covered by extremes of the spike trains, e.g. $[\min(t^{(1)}, \bar{t}^{(1)}), \; \max(t^{(n)}, \bar{t}^{(\bar{n})})]$ for a pair of spike trains, the result of the integration in the area not covered by extremes of the spike trains adds to the distance without contributing information about the spike trains. Thus, the integrating interval should preferably be chosen as the interval covered by the extremes of the considered set of spike trains. Alternatively, one may artificially add to all considered spike trains two extra spikes at the extremities $a$ and $b$ of the integrating interval, a procedure that is also used by the Kreuz et al. metrics \cite{kreuz2013}.

We have introduced localized versions of our metrics, which, depending on the localization kernel  $\mathcal{L}$, could have a biologically-relevant interpretation of measuring the differences between two spike trains as they are perceived at a particular moment in time by a neuron receiving these spike trains.

\section{Conclusion}
We have introduced here a new class of spike train metrics, inspired by the Pompeiu-Hausdorff distance. The max-metric and the modulus-metric behave in a qualitatively different way than classical metrics and are particularly suitable for measuring distances in spike trains where information is encoded in the identity of bursts, as unitary events. The modulus-metric does not depend on any parameters and can be computed faster than other distances, in a time that depends linearly on the number of spikes in the compared spike trains.

\subsection*{Acknowledgements}
We acknowledge the useful suggestions of Adriana Nicolae and Ovidiu Jurju\c{t}. This work was funded by the Sectorial Operational Programme Human Resources Development (POSDRU, contract no. 6/1.5/S/3, "Doctoral studies: through science towards society"), by a grant of the Romanian National Authority for Scientific Research (PNCDI II, Parteneriate, contract no. 11-039/2007) and by the Max Planck - Coneural PartnerGroup.

\section*{Appendix}

\appendix
%%%%%%%%%%%%%%%%%%%%%%%%%%%%%%%%%%%%%%%%%%%%%%%%%%%%%%%%%%%%%%%%%%%%%%%%%%%%%%%%%%%%%%%%%%%%%%%%%%%%%%%%%%%%%%%%%%%%%%%%%
%
%EQUIVALENT HAUSDORFF FORM
%
%%%%%%%%%%%%%%%%%%%%%%%%%%%%%%%%%%%%%%%%%%%%%%%%%%%%%%%%%%%%%%%%%%%%%%%%%%%%%%%%%%%%%%%%%%%%%%%%%%%%%%%%%%%%%%%%%%%%%%%%%
\section{Equivalent Hausdorff form}\label{section_equiv_haus}
\begin{proposition}
The Pompeiu-Hausdorff metric
\begin{equation}\label{hausb}
h(T,\bar{T})  = \sup_{x \in \mathbb{R}} \left| d(x, T)-d(x, \bar{T}) \right|
\end{equation}
can be equivalently expressed as
\begin{equation}
 h(T,\bar{T}) = \sup_{x \in [a, \, b]}  \left| d(x,T) - d(x,\bar{T}) \right|.
 \end{equation}
\end{proposition}
\begin{proof}
The result follows by showing that the global supremum is achieved on the interval $[a,b]$.
We need to show that restricting the interval in Eq. \ref{hausb} is acceptable since the supremum is achieved on the interval $[a,b]$.

We first consider $x\in (-\infty, a]$. Because $\forall \; t \in T$, $x \leq a \leq  t^{(1)} \leq t$, we have $d(x, T)=t^{(1)} -x$ and $d(a, T)=t^{(1)} -a$. Because $\forall \; \bar{t} \in \bar{T}$, $x \leq a \leq \bar{t}^{(1)} \leq \bar{t}$, we have $d(x, \bar{T})=\bar{t}^{(1)}-x$ and $d(a, \bar{T})=\bar{t}^{(1)}-a$. Thus
\begin{align}
\sup_{x \in (-\infty, a]}  \left| d(x,T) - d(x,\bar{T}) \right|&=\sup_{x \in (-\infty, a]}  \left| (t^{(1)}-x) - (\bar{t}^{(1)}-x) \right| = \left|  t^{(1)}-\bar{t}^{(1)} \right|\\
\sup_{x \in (-\infty, a]}  \left| d(x,T) - d(x,\bar{T}) \right|&=\left| d(a,T) - d(a,\bar{T}) \right|. \label{equation_equivalent_1}
\end{align}

Analogously, for $x \in [b, \infty)$ we have
\begin{equation}
\sup_{x \in [b, \infty)}  \left| d(x,T) - d(x,\bar{T}) \right|=\left| d(b,T) - d(b,\bar{T}) \right|.\label{equation_equivalent_2}
\end{equation}

From Eqs. \ref{equation_equivalent_1} and \ref{equation_equivalent_2},
\begin{equation}
\sup_{x \in \mathbb{R}}  \left| d(x,T) - d(x,\bar{T}) \right|=\sup_{x \in [a, \, b]}  \left| d(x,T) - d(x,\bar{T}) \right|.
\end{equation}
\end{proof}
%%%%%%%%%%%%%%%%%%%%%%%%%%%%%%%%%%%%%%%%%%%%%%%%%%%%%%%%%%%%%%%%%%%%%%%%%%%%%%%%%%%%%%%%%%%%%%%%%%%%%%%%%%%%%%%%%%%%%%%%%
%
%MAX METRIC
%
%%%%%%%%%%%%%%%%%%%%%%%%%%%%%%%%%%%%%%%%%%%%%%%%%%%%%%%%%%%%%%%%%%%%%%%%%%%%%%%%%%%%%%%%%%%%%%%%%%%%%%%%%%%%%%%%%%%%%%%%%

\section{Analysis of the max-metric}\label{section_analysis_max_metric}

%%%%%%%%%%%%%%%%%%%%%%%%%%%%%%%%%%%%%%%%%%%%%%%%%%%%%%%%%%%%%%%%%%%%%%%%%%%%%%%%%%%%%%%%%%%%%%%%%%%%%%%%%%%%%%%%%%%%%%%%%
%
%PROPOSITION 1 - MAX METRIC
%
%%%%%%%%%%%%%%%%%%%%%%%%%%%%%%%%%%%%%%%%%%%%%%%%%%%%%%%%%%%%%%%%%%%%%%%%%%%%%%%%%%%%%%%%%%%%%%%%%%%%%%%%%%%%%%%%%%%%%%%%%
\begin{proposition}
$d_m(T,\bar{T}) < \infty$.
\end{proposition}
\begin{proof}
\noindent From Eq. \ref{equation_hausdorff_ab}, for every $x \in [a,b]$ we have
\begin{equation}
\left|d(x,T) - d(x,\bar{T})\right| \le h(T,\bar{T}).
\end{equation}
Multiplying by $\mathcal{H}(\cdot)$, which is positive, we obtain, $\forall s \in [a,b]$,
\begin{equation}
\left|d(x,T) - d(x,\bar{T})\right| \; \mathcal{H}(|s-x|) \le h(T,\bar{T}) \; \mathcal{H}(|s-x|).
\end{equation}
By taking the supremum and integrating, the equation above becomes
\begin{equation}
\int_{a}^{b} \sup_{x\in [a,b]} \left\lbrace \left|d(t,T) - d(t,\bar{T})\right| \; \mathcal{H}(|s-x|) \right\rbrace \textrm{d}s \le \int_{a}^{b} \sup_{x\in [a,b]} \left\lbrace h(T,\bar{T}) \; \mathcal{H}(|s-x|) \right\rbrace \textrm{d}s.
\end{equation}
The left side of the equation above is the max-metric $d_m$ (Eq. \ref{maxm}). Because $h(T,\bar{T})$ is independent of $s$ and $x$,
\begin{equation}
 d_m(T,\bar{T}) \le h(T,\bar{T})\int_{a}^{b} \sup_{x\in [a,b]}  \mathcal{H}(|s-x|) \textrm{d}s.
\end{equation}
Since $h(T,\bar{T}) \le b-a$ and $\mathcal{H}(y) < m$, $\forall y \in [0, b-a]$ (Eq. \ref{equation_h_conditions}), we obtain
\begin{equation}
 d_m(T,\bar{T}) \le (b-a)^2 \; m < \infty.
\end{equation}
\end{proof}

%%%%%%%%%%%%%%%%%%%%%%%%%%%%%%%%%%%%%%%%%%%%%%%%%%%%%%%%%%%%%%%%%%%%%%%%%%%%%%%%%%%%%%%%%%%%%%%%%%%%%%%%%%%%%%%%%%%%%%%%%
%
%PROPOSITION 2 - MAX METRIC
%
%%%%%%%%%%%%%%%%%%%%%%%%%%%%%%%%%%%%%%%%%%%%%%%%%%%%%%%%%%%%%%%%%%%%%%%%%%%%%%%%%%%%%%%%%%%%%%%%%%%%%%%%%%%%%%%%%%%%%%%%%
\begin{proposition}
$d_m : \mathcal{S}_{[a,b]} \times \mathcal{S}_{[a,b]} \to \mathbb{R}^+$ is a metric.
\end{proposition}
\begin{proof}
In order to show that $d_m(\cdot,\cdot)$ is a metric we need to prove that it is non-negative, that $d_m(T,\bar{T}) = 0 \Leftrightarrow T = \bar{T}$ for any $T,\bar{T}\in \mathcal{S}_{[a,b]}$, that is symmetric, and that it satisfies the triangle inequality.

It is trivial to show that $d_m(T,\bar{T})$ is non-negative, symmetric, and that if $T = \bar{T} \Rightarrow d_m(T,\bar{T}) = 0$.

In order to prove that $d_m(T,\bar{T}) = 0 \Rightarrow T = \bar{T}$ we use a \textit{reductio ad absurdum} argument. Assume  that $d_m(T,\bar{T}) = 0$ with $T \neq \bar{T}$. Then, there must be at least one spike in one of the two spike trains that is not in the other, since we do not allow overlapping spikes within the spike trains. Consider that this spike belongs to $T$; in the case that it belongs to $\bar{T}$, the proof is analogous. Let $u$ be the timing of this spike, $u \in T \setminus \bar{T}$. Because $u \in [a,b]$, we have, $\forall s \in [a,b]$,
\begin{equation}
\sup_{x \in {[a,b]}} \left\lbrace  \left|d(x,T) - d(x,\bar{T})\right| \; \mathcal{H}(|s-x|)\right\rbrace \ge \left|d(u,T) - d(u,\bar{T})\right| \; \mathcal{H}(|s-u|).
\end{equation}
Because $u \in T$, $d(u,T) = 0$. Thus
\begin{equation}
\sup_{x \in {[a,b]}} \left\lbrace  \left|d(x,T) - d(x,\bar{T})\right| \; \mathcal{H}(|s-x|)\right\rbrace \ge d(u,\bar{T}) \; \mathcal{H}(|s-u|).
\end{equation}
By integrating the equation above, we obtain
\begin{equation}\label{star1}
d_m(T,\bar{T}) \ge d(u,\bar{T}) \; \int_{a}^{b} \mathcal{H}(|s-u|) \; \textrm{d}s.
\end{equation}
Because $u \notin \bar{T}$, $d(u,\bar{T}) > 0$. Also considering Eq. \ref{equation_h_condition_integral}, we have
\begin{equation}\label{star67574}
d(u,\bar{T}) \; \int_{a}^{b} \mathcal{H}(|s-u|) \; \textrm{d}s >0.
\end{equation}
Thus, from Eqs. \ref{star1} and \ref{star67574}, $d_m(T,\bar{T})>0$. Since we considered that $d_m(T,\bar{T}) = 0$, this cannot be true.  Hence, $T \subseteq \bar{T}$. Likewise, one can show that $\bar{T} \subseteq T$ and so $T = \bar{T}$.

In order to prove the triangle inequality, consider $\hat{T} \in \mathcal{S}_{[a,b]}$. We have, $\forall x \in \mathbb{R}$,
\begin{equation}
| d(x,T) - d(x,\hat{T}) + d(x,\hat{T}) -d(x,\bar{T})| \leq  |d(x,T) - d(x,\hat{T}) | + |d(x,\hat{T}) -d(x,\bar{T})|.
\end{equation}
Because $\sup_x(f(x) + g(x)) \leq \sup_x(f(x)) + \sup_x(g(x))$ for any two functions $f$ and $g$, it follows that
\begin{align}
\sup_{x \in [a, b] } \left\lbrace  \left|d(x,T) - d(x,\bar{T})\right| \; \mathcal{H}(|s-x|)\right\rbrace &  \le \sup_{x \in [a, b]} \left\lbrace  \left|d(x,T) - d(x,\hat{T})\right| \; \mathcal{H}(|s-x|)\right\rbrace + \nonumber \\
&\quad \sup_{x \in [a, b]} \left\lbrace  \left|d(x,\hat{T}) - d(x,\bar{T})\right| \; \mathcal{H}(|s-x|)\right\rbrace.
\end{align}
After integration, it results that
\begin{equation}
d_m(T,\bar{T}) \le d_m(T,\hat{T}) + d_m(\hat{T},\bar{T}).
\end{equation}

With this final equality we have shown that the distance is indeed a metric and the proof ends.
\end{proof}

%%%%%%%%%%%%%%%%%%%%%%%%%%%%%%%%%%%%%%%%%%%%%%%%%%%%%%%%%%%%%%%%%%%%%%%%%%%%%%%%%%%%%%%%%%%%%%%%%%%%%%%%%%%%%%%%%%%%%%%%%
%
%PROPOSITION 3 - MAX METRIC
%
%%%%%%%%%%%%%%%%%%%%%%%%%%%%%%%%%%%%%%%%%%%%%%%%%%%%%%%%%%%%%%%%%%%%%%%%%%%%%%%%%%%%%%%%%%%%%%%%%%%%%%%%%%%%%%%%%%%%%%%%%
\begin{proposition}\label{echiv1} The metric $d_m : \mathcal{S}_{[a, b]} \times \mathcal{S}_{[a, b]} \to \mathbb{R}$ is topologically equivalent to the Pompeiu-Hausdorff distance.

\end{proposition}

\begin{proof}

In order to show that the metrics $d_m$ and $h$ are topologically equivalent it is sufficient to prove that the identity function  $i_{\mathcal{S}_{[a, b]}} : (\mathcal{S}_{[a, b]},d_m)\rightarrow(\mathcal{S}_{[a, b]},h)$, $i_{\mathcal{S}_{[a, b]}}(T) = T$  and its inverse are both continuous (\citealp[p. 229]{osearcoid}; \citealp[p. 12]{DD09}).

We first show that $i_{\mathcal{S}_{[a, b]}}$ is continuous, which is equivalent to:
$\forall \; T \in \mathcal{S}_{[a, b]}, \; \forall \; \epsilon > 0,  \; \exists \; \delta(\epsilon) > 0 \; \; \text{such that} \;\; \forall \; \bar{T} \in \mathcal{S}_{[a, b]} \; \text{with} \; \; d_m(T,\bar{T}) \leq \delta \;\; \text{we have} \;\; h(T,\bar{T}) \leq \epsilon$.
We choose $\delta(\epsilon) = \epsilon \; A(T,\bar{T})$ with
\begin{equation}
A(T,\bar{T}) = \inf_{t \in T \cup \bar{T}} \int_{a}^{b}\mathcal{H}(|s-t|) \; \mbox{d}s.
\end{equation}
From Eq. \ref{equation_h_condition_integral}, we have that $A(T,\bar{T}) > 0$. From Eq. \ref{maxm}, for all $t \in T$,
\begin{align}
d_m(T,\bar{T}) & \ge \int_{a}^{b} |d(t, T) - d(t,\bar{T})| \; \mathcal{H}(|s-t|) \mbox{d}s \\
&=\int_{a}^{b} d(t,\bar{T}) \; \mathcal{H}(|s-t|) \mbox{d}s \\
& \ge d(t,\bar{T}) \inf_{t' \in T} \int_{a}^{b} \mathcal{H}(|s-t'|) \; \mbox{d}s
\end{align}
and thus
\begin{align}
d_m(T,\bar{T}) \ge & \sup_{t \in T} \, d(t,\bar{T}) \; \inf_{t' \in T} \int_{a}^{b} \mathcal{H}(|s-t'|) \; \mbox{d}s \\
d_m(T,\bar{T}) \ge  & \sup_{t \in T} \, d(t,\bar{T}) \; \inf_{t' \in T \cup \bar{T}} \int_{a}^{b} \mathcal{H}(|s-t'|) \; \mbox{d}s \label{m1p}.
\end{align}
Analogously, for all $\bar{t} \in \bar{T}$,
\begin{equation}\label{m2p}
d_m(T,\bar{T}) \ge \sup_{\bar{t} \in \bar{T}} \, d(\bar{t},T) \; \inf_{t' \in T \cup \bar{T}} \int_{a}^{b} \mathcal{H}(|s-t'|) \; \mbox{d}s.
\end{equation}
Taking the max value in Eqs. \ref{m1p} and  \ref{m2p} we obtain
\begin{equation}
d_m(T,\bar{T}) \ge \max \left\lbrace \sup_{t \in T} \, d(t,\bar{T}), \; \sup_{\bar{t} \in \bar{T}} \, d(\bar{t},T) \right\rbrace \\ \inf_{t' \in T \cup \bar{T}} \int_a^b \mathcal{H}(|s-t'|) \; \mbox{d}s.
\end{equation}
From Eq. \ref{equation_hausdorffii} it follows that
\begin{equation}
d_m(T,\bar{T}) \ge h(T,\bar{T}) \inf_{t \in T \cup \bar{T}} \int_{a}^{b} \mathcal{H}(|s-t|) \; \mbox{d}s.
\end{equation}
Because $\epsilon \; A(T,\bar{T}) = \delta(\epsilon) \geq  d_m(T,\bar{T})$, from the last equation we have that $\epsilon \; A(T,\bar{T}) \geq h(T,\bar{T}) \; A(T,\bar{T})$. Thus $h(T,\bar{T}) \leq \epsilon$ since $A(T,\bar{T})>0$.

It remains to show that $i_{\mathcal{S}_{[a, b]}}^{-1}$ is continuous, which is equivalent to:
$\forall \; T \in \mathcal{S}_{[a, b]}, \; \forall \; \epsilon > 0,  \; \exists \; \delta(\epsilon) > 0 \; \; \text{such that} \;\; \forall \; \bar{T} \in \mathcal{S}_{[a, b]} \; \text{with} \; \; h(T,\bar{T}) \leq \delta \;\; \text{we have} \;\; d_m(T,\bar{T}) \leq \epsilon$.
We choose $\delta(\epsilon) = \epsilon / B(a,b)$ with
\begin{equation}
B(a,b) =  \int_{a}^{b}\sup_{x\in[a,b]} \mathcal{H}(|s-x|) \; \mbox{d}s.
\end{equation}
We have
\begin{equation}
B(a,b) \geq  \sup_{x\in[a,b]} \int_{a}^{b} \mathcal{H}(|s-x|) \; \mbox{d}s.\label{equation_b}
\end{equation}
Considering Eq. \ref{equation_h_condition_integral}, we have $B(a,b)  > 0$. From Eq. \ref{maxm} we have
\begin{align}
 d_m(T,\bar{T})\le & \int_{a}^{b} \sup_{x \in [a,b]} \left|d(x,T) - d(x,\bar{T})\right| \; \sup_{x \in [a,b]} \mathcal{H}(|s-x|) \; \mbox{d}s\\
 d_m(T,\bar{T})\le & \sup_{x \in [a,b]} \left|d(x,T) - d(x,\bar{T})\right| \int_{a}^{b} \sup_{x \in [a,b]} \mathcal{H}(|s-x|) \; \mbox{d}s.
\end{align}
From Eqs. \ref{equation_hausdorff_ab} and \ref{equation_b} we have
\begin{equation}
 d_m(T,\bar{T}) \leq h(T,\bar{T}) \; B(a,b).
\end{equation}
Since $h(T,\bar{T}) \leq \delta(\epsilon) = \epsilon / B(a,b)$ and $B(a,b)>0$, we get $h(T,\bar{T}) \; B(a,b)\leq \epsilon$ and, finally, $d_m(T,\bar{T}) \leq \epsilon$.
\end{proof}

%%%%%%%%%%%%%%%%%%%%%%%%%%%%%%%%%%%%%%%%%%%%%%%%%%%%%%%%%%%%%%%%%%%%%%%%%%%%%%%%%%%%%%%%%%%%%%%%%%%%%%%%%%%%%%%%%%%%%%%%%
%
%CONVOLUTION MAX METRIC
%
%%%%%%%%%%%%%%%%%%%%%%%%%%%%%%%%%%%%%%%%%%%%%%%%%%%%%%%%%%%%%%%%%%%%%%%%%%%%%%%%%%%%%%%%%%%%%%%%%%%%%%%%%%%%%%%%%%%%%%%%%

\section{Analysis of the convolution max-metric} \label{section_analysis_convolution_max_metric}
%%%%%%%%%%%%%%%%%%%%%%%%%%%%%%%%%%%%%%%%%%%%%%%%%%%%%%%%%%%%%%%%%%%%%%%%%%%%%%%%%%%%%%%%%%%%%%%%%%%%%%%%%%%%%%%%%%%%%%%%%
%
%LEMMA 1 - CONVOLUTION MAX METRIC
%
%%%%%%%%%%%%%%%%%%%%%%%%%%%%%%%%%%%%%%%%%%%%%%%%%%%%%%%%%%%%%%%%%%%%%%%%%%%%%%%%%%%%%%%%%%%%%%%%%%%%%%%%%%%%%%%%%%%%%%%%%
\begin{lemma}\label{l1}
Let $g \colon [a,b] \rightarrow \mathbb{R}$ be a continuous function and $h \colon [0,b-a] \rightarrow \mathbb{R}$ be a continuous function which is derivable on $(0, b-a)$ and has bounded derivative. Then the function $q \colon [a,b] \rightarrow \mathbb{R}$,
\begin{equation}
q(s) = \sup_{x \in [a,b]} \; \left[ g(x) \; h(|s-x|) \right]
\end{equation}
is continuous on $[a,b]$.
\end{lemma}

\begin{proof}
Consider $s_0 \in [a,b]$. We need to show that $\forall \; \epsilon > 0$, $\exists \; \delta(\epsilon) > 0$ such that $\forall \; s \in (s_0 - \delta, s_0 + \delta) \bigcap [a, b]$, we have
\begin{equation}
\left| \sup_{x \in [a,b]} \left[ g(x) \; h(|s-x|)\right] - \sup_{x \in [a,b]} \left[ g(x) \; h(|s_0-x|) \right] \right| < \epsilon. \label{equation_lemma_to_demonstrate}
\end{equation}

We have
\begin{align}
g(x) \; h(|s-x|) & -g(x) \; h(|s_0-x|)\le  \left| g(x) \; h(|s-x|) - g(x) \; h(|s_0-x|) \right| \label{equation_lemma_s}\\
g(x) \; h(|s-x|)&\le  \left| g(x) \; [ h(|s-x|) - h(|s_0-x|)] \right| + g(x) \; h(|s_0-x|)  \\
g(x) \; h(|s-x|) &\le |g(x)| \; \left| h(|s-x|) - h(|s_0-x|) \right|+g(x) \; h(|s_0-x|) \label{equation_lemma_a}.
\end{align}
The function $g$ is bounded since it is continuous on a compact interval \citep[p. 56]{Protter}. We denote by $M$ the bound of the absolute value of $g$, i.e. $|g(x)| \le M$, $\forall x \in [a,b]$. We denote by $L$ the bound of the absolute value of the derivative of $h$, i.e. $|h(x)-h(y)|\le L \; |x-y|$, $\forall \; x,y \in [0,b-a]$. Let $\epsilon > 0$ and $\delta = \epsilon/(M \; L) $. Then for all $s \in (s_0 -\delta, s_0 + \delta) \bigcap [a, b]$,
\begin{align}
\left| h(|s-x|) - h(|s_0-x|) \right| &\le L \; \left| |s-x| - |s_0-x| \right| \\
|g(x)| \left| h(|s-x|) - h(|s_0-x|) \right|& \le M \; L \; \left| |s-x| - |s_0-x| \right| .\label{equation_lemma_b}
\end{align}
From Eqs. \ref{equation_lemma_a} and \ref{equation_lemma_b},
\begin{equation}
g(x) \; h(|s-x|)  \le M \; L \; \left| |s-x| - |s_0-x| \right| + g(x) \; h(|s_0-x|).
\end{equation}
Because $\forall u,v \in \mathbb{R}$ we have $| |u| - |v| | \leq |u-v|$, it follows that
\begin{equation}
g(x) \; h(|s-x|) \le  M \; L \; |s - s_0| + g(x) \;  h(|s_0-x|).
\end{equation}
Because $|s - s_0| < \delta = \epsilon/(M \; L)$,
\begin{align}
g(x) \; h(|s-x|)& < M \; L \; \frac{\epsilon}{M \;  L} + g(x) \;  h(|s_0-x|) \nonumber \\
&=\epsilon + g(x) \;  h(|s_0-x|).
\end{align}
Applying supremum to the equation above we obtain
\begin{equation}
\sup_{x \in [a,b]} \left[ g(x) \; h(|s-x|) \right] < \epsilon + \sup_{x \in [a,b]} \left[ g(x) \; h(|s_0-x|) \right].
\end{equation}
It follows that
\begin{equation}
\sup_{x \in [a,b]} \left[ g(x) \; h(|s-x|) \right] - \sup_{x \in [a,b]} \left[ g(x) \; h(|s_0-x|) \right] < \epsilon.
\end{equation}
Analogously, by switching $s$ and $s_0$ in Eq. \ref{equation_lemma_s} and the ensuing equations, we get
\begin{equation}
\sup_{x \in [a,b]} \left[ g(x) \; h(|s_0-x|) \right] - \sup_{x \in [a,b]} \left[ g(x) \; h(|s-x|) \right] < \epsilon.
\end{equation}
Thus we have proved Eq.  \ref{equation_lemma_to_demonstrate} and the proof ends.
\end{proof}

%%%%%%%%%%%%%%%%%%%%%%%%%%%%%%%%%%%%%%%%%%%%%%%%%%%%%%%%%%%%%%%%%%%%%%%%%%%%%%%%%%%%%%%%%%%%%%%%%%%%%%%%%%%%%%%%%%%%%%%%%
%
%PROPOSITION 1 - CONVOLUTION MAX METRIC
%
%%%%%%%%%%%%%%%%%%%%%%%%%%%%%%%%%%%%%%%%%%%%%%%%%%%%%%%%%%%%%%%%%%%%%%%%%%%%%%%%%%%%%%%%%%%%%%%%%%%%%%%%%%%%%%%%%%%%%%%%%
\begin{proposition}
$d_c(T,\bar{T}) < \infty$.
\end{proposition}
\begin{proof}
\noindent Because $\mathcal{K}$ is positive, we have $f \geq 0$ and $\bar{f} \geq 0$. Moreover, $\forall x \in \mathbb{R}$,
\begin{align}
 \left| f(x) - \bar{f}(x) \right| &\le  \left| f(x) \right| + \left|  \bar{f}(x) \right| = f(x) + \bar{f}(x) \nonumber \\
   \left| f(x) - \bar{f}(x) \right| &\le\sum_{i=1}^{n} \mathcal{K}(x-t^{(i)}) + \sum_{i=1}^{\bar{n}} \mathcal{K}(x-\bar{t}^{(i)}).
\end{align}
Since $\mathcal{K}(x) \le p$, $\forall \; x \in \mathbb{R}$,
\begin{equation}
\left| f(x) - \bar{f}(x) \right| \le p\; (n + \bar{n}).
\end{equation}
From Eq. \ref{m1},
\begin{align}
 d_c(T,\bar{T}) \le & \int_{a}^{b} \sup_{x\in[a,b]} \left| f(x) - \bar{f}(x) \right| \sup_{x\in[a,b]} \mathcal{H}(|s-x|) \; \textrm{d}s\\
 d_c(T,\bar{T}) \le  & p\; (n + \bar{n}) \int_{a}^{b} \sup_{x\in[a,b]} \mathcal{H}(|s-x|) \; \textrm{d}s.
\end{align}
Because $\mathcal{H} (y) <m, \; \forall \; y \in [0, b-a] $ (Eq. \ref{equation_h_conditions}), it follows that
\begin{equation}
 d_c(T,\bar{T}) \le p\; (n + \bar{n})\;  (b-a) \; m < \infty.
\end{equation}

\end{proof}

%%%%%%%%%%%%%%%%%%%%%%%%%%%%%%%%%%%%%%%%%%%%%%%%%%%%%%%%%%%%%%%%%%%%%%%%%%%%%%%%%%%%%%%%%%%%%%%%%%%%%%%%%%%%%%%%%%%%%%%%%
%
%PROPOSITION 2 - CONVOLUTION MAX METRIC
%
%%%%%%%%%%%%%%%%%%%%%%%%%%%%%%%%%%%%%%%%%%%%%%%%%%%%%%%%%%%%%%%%%%%%%%%%%%%%%%%%%%%%%%%%%%%%%%%%%%%%%%%%%%%%%%%%%%%%%%%%%
\begin{proposition}
$d_c : \mathcal{S}_{[a, b]} \times \mathcal{S}_{[a, b]} \to \mathbb{R}^+$ is a metric.
\end{proposition}
\begin{proof}
In order to show that $d_c(\cdot,\cdot)$ is a metric we need to prove that it is non-negative, that $d_c(T,\bar{T}) = 0 \Leftrightarrow T = \bar{T}$ for any $T,\bar{T}\in \mathcal{S}_{[a, b]}$, that it is symmetric, and that it satisfies the triangle inequality.

It is trivial to show that $d_c(T,\bar{T})$ is non-negative, symmetric, and that if $T = \bar{T} \Rightarrow d_c(T,\bar{T}) = 0$.
In order to prove that $d_c(T,\bar{T}) = 0 \Rightarrow T = \bar{T}$  we use a \textit{reductio ad absurdum} argument. Assume  that $d_c(T,\bar{T}) = 0$ with $T \neq \bar{T}$.
For $s \in [a,b]$ let
\begin{equation}
q(s) = \sup_{x \in [a,b]} \left\lbrace \left| f(x) - \bar{f}(x) \right| \; \mathcal{H}(|s-x|) \right\rbrace.
\end{equation}
Because $|f(x) - \bar{f}(x)|$ is continuous, from the properties of $\mathcal{H}$ and Lemma \ref{l1} we obtain that $q$ is continuous. Because of the properties of $\mathcal{K}$, $T \neq \bar{T} \Rightarrow \exists \;  x \in [a,b]$ such that $f(x) \neq \bar{f}(x)$; and  because $\mathcal{H}$ is strictly positive,  it follows that $q$ is not zero everywhere, $\exists x \in [a,b]$ such that $q(x) > 0$. Because $q$ is continuous, it follows that
\begin{equation}
d_c(T,\bar{T}) = \int_{a}^{b} q(s) \; \textrm{d}s > 0
\end{equation}
which contradicts the hypothesis that $d_c(T,\bar{T}) = 0$. Hence, $T = \bar{T}$.

In order to prove the triangle inequality, consider $\hat{T} \in \mathcal{S}_{[a, b]}$. We have, $\forall x, s \in [a, b]$,
\begin{equation}
\left| f(x) - \hat{f}(x) + \hat{f}(x) - \bar{f}(x) \right| \; \mathcal{H}(|s-x|) \le \left( \left| f(x) - \hat{f}(x) \right| + \left| \hat{f}(x) - \bar{f}(x) \right| \right) \; \mathcal{H}(|s-x|).
\end{equation}
Because $\sup_s(g(s) + h(s)) \leq \sup_s(g(s)) + \sup_s(h(s))$ for any two functions $g$ and $h$, it follows that, $\forall s \in [a,b]$,
\begin{align}
\sup_{x \in {[a,b]} } \left\lbrace  \left| f(x) - \bar{f}(x) \right| \; \mathcal{H}(|s-x|)\right\rbrace &  \le \sup_{x \in {[a,b]} } \left\lbrace  \left|f(x) - \hat{f}(x)\right| \; \mathcal{H}(|s-x|)\right\rbrace + \nonumber \\
&\quad \sup_{x \in {[a,b]} } \left\lbrace  \left| \hat{f}(x) - \bar{f}(x)\right| \; \mathcal{H}(|s-x|)\right\rbrace.
\end{align}
After integration, it results that
\begin{equation}
d_c(T,\bar{T}) \le d_c(T,\hat{T}) + d_c(\hat{T},\bar{T}).
\end{equation}

With this final equality we have shown that the distance is indeed a metric and the proof ends.
\end{proof}

%%%%%%%%%%%%%%%%%%%%%%%%%%%%%%%%%%%%%%%%%%%%%%%%%%%%%%%%%%%%%%%%%%%%%%%%%%%%%%%%%%%%%%%%%%%%%%%%%%%%%%%%%%%%%%%%%%%%%%%%%
%
%LOCALIZED MAX METRIC
%
%%%%%%%%%%%%%%%%%%%%%%%%%%%%%%%%%%%%%%%%%%%%%%%%%%%%%%%%%%%%%%%%%%%%%%%%%%%%%%%%%%%%%%%%%%%%%%%%%%%%%%%%%%%%%%%%%%%%%%%%%

\section{Analysis of the localized max-metric}\label{section_analysis_localized_max_metric}

%%%%%%%%%%%%%%%%%%%%%%%%%%%%%%%%%%%%%%%%%%%%%%%%%%%%%%%%%%%%%%%%%%%%%%%%%%%%%%%%%%%%%%%%%%%%%%%%%%%%%%%%%%%%%%%%%%%%%%%%%
%
%PROPOSITION 1 - LOCALIZED MAX METRIC
%
%%%%%%%%%%%%%%%%%%%%%%%%%%%%%%%%%%%%%%%%%%%%%%%%%%%%%%%%%%%%%%%%%%%%%%%%%%%%%%%%%%%%%%%%%%%%%%%%%%%%%%%%%%%%%%%%%%%%%%%%%
\begin{proposition}
$d_l(T,\bar{T}) < \infty$.
\end{proposition}
\begin{proof}
\noindent For every $s \in [a,b]$ we have
\begin{align}
 h(T,\bar{T}) &= \sup_{x\in[a,b]}  \left| d(x,T) - d(x,\bar{T}) \right|  \nonumber \\
  &\ge \sup_{x\in[s,b]}  \left| d(x,T) - d(x,\bar{T}) \right| \\
  h(T,\bar{T}) \; \mathcal{L}(b-s) &\ge \mathcal{L}(b-s)\sup_{x\in[s,b]}  \left| d(x,T) - d(x,\bar{T}) \right| .
\end{align}
By integrating the last equation above, we obtain
\begin{equation}
\int_{a}^{b} \mathcal{L}(b-s)\sup_{x\in[s,b]} \left| d(x,T) - d(x,\bar{T}) \right|  \textrm{d}s \leq h(T,\bar{T}) \int_{a}^{b} \mathcal{L}(b-s) \textrm{d}s.
\end{equation}
Since $h(T,\bar{T}) \le b-a$ and $\mathcal{L}(y) < m$, $\forall y \in [0, b-a]$ (Eq. \ref{equation_l_conditions}), it follows that
\begin{equation}
 d_l(T,\bar{T}) \le m\; (b-a)^2 < \infty.
\end{equation}
\end{proof}

%%%%%%%%%%%%%%%%%%%%%%%%%%%%%%%%%%%%%%%%%%%%%%%%%%%%%%%%%%%%%%%%%%%%%%%%%%%%%%%%%%%%%%%%%%%%%%%%%%%%%%%%%%%%%%%%%%%%%%%%%
%
%PROPOSITION 2 - LOCALIZED MAX METRIC
%
%%%%%%%%%%%%%%%%%%%%%%%%%%%%%%%%%%%%%%%%%%%%%%%%%%%%%%%%%%%%%%%%%%%%%%%%%%%%%%%%%%%%%%%%%%%%%%%%%%%%%%%%%%%%%%%%%%%%%%%%%
\begin{proposition}
$d_l : \mathcal{S}_{[a,b]} \times \mathcal{S}_{[a,b]} \to \mathbb{R}^+$ is a metric.
\end{proposition}
\begin{proof}

In order to show that $d_l(\cdot,\cdot)$ is a metric we need to prove that it is non-negative, that $d_l(T,\bar{T}) = 0 \Leftrightarrow T = \bar{T}$ for any $T,\bar{T}\in \mathcal{S}_{[a,b]}$, that it is symmetric, and that it satisfies the triangle inequality.

It is trivial to show that $d_l(T,\bar{T})$ is non-negative, symmetric, and that if $T = \bar{T} \Rightarrow d_l(T,\bar{T}) = 0$.

In order to prove that $d_l(T,\bar{T}) = 0 \Rightarrow T = \bar{T}$ we use a \textit{reductio ad absurdum} argument. Assume  that $d_l(T,\bar{T}) = 0$ with $T \neq \bar{T}$. Then, there must be at least one spike in one of the two spike trains that is not in the other. Consider that this spike belongs to $T$; in the case it belongs to $\bar{T}$, the proof is analogous. Let $u$ be the timing of this spike, $u \in T \setminus \bar{T}$.

First, we consider the case $u>a$. We have, $\forall s \in [a,u]$,
\begin{equation}
\sup_{x \in {[s,b]}}   \left|d(x,T) - d(x,\bar{T})\right|\ge \left|d(u,T) - d(u,\bar{T})\right|.
\end{equation}
Because $u \in T \setminus \bar{T}$, $d(u,T) = 0$ and $d(u,\bar{T}) > 0$. Thus,
\begin{equation}
\sup_{x \in {[s,b]}}  \left|d(x,T) - d(x,\bar{T})\right|\ge d(u,\bar{T}).
\end{equation}
Multiplying by $\mathcal{L}(\cdot)$ and integrating the above equation, we obtain
\begin{equation}\label{star}
\int_{a}^{u} \mathcal{L}(b-s) \sup_{x \in {[s,b]}}   \left|d(x,T) - d(x,\bar{T})\right| \; \textrm{d}s \ge d(u,\bar{T}) \int_{a}^{u} \mathcal{L}(b-s)\; \textrm{d}s.
\end{equation}
We also have
\begin{equation}
\int_{u}^{b} \mathcal{L}(b-s) \sup_{x \in {[s,b]}}   \left|d(x,T) - d(x,\bar{T})\right| \; \textrm{d}s \geq 0.
\end{equation}
By adding the last two inequations above, and considering Eq. \ref{m4}, we have:
\begin{equation}
d_l(T,\bar{T}) \ge  d(u,\bar{T}) \int_{a}^{u} \mathcal{L}(b-s)\; \textrm{d}s.
\end{equation}
Because $\mathcal{L}$ is strictly positive on $(0, b-a]$ and continuous, we have $\int_{a}^{u} \mathcal{L}(b-s)\; \textrm{d}s > 0$. Because we have $d(u,\bar{T})>0$, we get
\begin{equation}
d_l(T,\bar{T}) >0.\label{equation_localized_max_metric_proof_1}
\end{equation}
Since we have considered $d_l(T,\bar{T}) = 0$, it follows that Eq. \ref{equation_localized_max_metric_proof_1} cannot be true.  Hence, $T \subseteq \bar{T}$.

Second, we consider the case $u=a$. Let $v$ be the timing of the first spike in either $T$ or $\bar{T}$, other than $u$. Since  $T \neq \bar{T}$, $v>u$. Because
\begin{equation}
\int_{(u+v)/2}^{b} \mathcal{L}(b-s) \sup_{x \in {[s,b]}}   \left|d(x,T) - d(x,\bar{T})\right| \; \textrm{d}s \geq 0,
\end{equation}
from Eq. \ref{m4} we have
\begin{align}
d_l(T,\bar{T}) \ge &\int_{a}^{(u+v)/2} \mathcal{L}(b-s) \sup_{x \in {[s,b]}}   \left|d(x,T) - d(x,\bar{T})\right| \; \textrm{d}s \\
d_l(T,\bar{T}) \ge & \int_{a}^{(u+v)/2} \mathcal{L}(b-s) \sup_{x \in {[s,(u+v)/2]}}   \left|d(x,T) - d(x,\bar{T})\right| \; \textrm{d}s .
\end{align}
For all $x \in \left[u, (u+v)/2\right)$ we have $d(x,T)=x-u$, $d(x,\bar{T}) \geq v-x$, and because on this interval $v-x>x-u$, we have $\left|d(x,T) - d(x,\bar{T})\right|>0$. Because $\mathcal{L}$ is strictly positive on $(0, b-a]$, we get
\begin{equation}
\int_{a}^{(u+v)/2} \mathcal{L}(b-s) \sup_{x \in {[s,(u+v)/2]}}   \left|d(x,T) - d(x,\bar{T})\right| \; \textrm{d}s >0
\end{equation}
and thus \begin{equation}
d_l(T,\bar{T}) >0.\label{equation_localized_max_metric_proof_1b}
\end{equation}
Since we have considered $d_l(T,\bar{T}) = 0$, it follows that Eq. \ref{equation_localized_max_metric_proof_1b} cannot be true.  Hence, $T \subseteq \bar{T}$.

Thus, we have shown that in both the case $u>a$ and the case $u=a$ we have $T \subseteq \bar{T}$. Likewise, one can show that $\bar{T} \subseteq T$ and so $T = \bar{T}$ if $d_l(T,\bar{T}) = 0$.

In order to prove the triangle inequality consider $\hat{T} \in \mathcal{S}_{[a, b]}$. We have, $\forall x \in [a,b]$,
\begin{equation}
| d(x,T) - d(x,\hat{T}) + d(x,\hat{T}) -d(x,\bar{T})| \leq  |d(x,T) - d(x,\hat{T}) | + |d(x,\hat{T}) -d(x,\bar{T})|.
\end{equation}
Because $\sup_x(g(x) + h(x)) \leq \sup_x(g(x)) + \sup_x(h(x))$ for any two functions $g$ and $h$, it follows that, $\forall s \in [a,b]$,
\begin{align}
\mathcal{L}(b-s) \sup_{x \in {[s,b]} }   \left|d(x,T) - d(x,\bar{T})\right| &  \le \mathcal{L}(b-s) \sup_{x \in {[s,b]} }  \left|d(x,T) - d(x,\hat{T})\right| + \nonumber \\
&\quad \mathcal{L}(b-s) \sup_{x \in {[s,b]} }   \left|d(x,\hat{T}) - d(x,\bar{T})\right|.
\end{align}
After integration, it results that $d_l(T,\bar{T}) \le d_l(T,\hat{T}) + d_l(\hat{T},\bar{T})$.

With this final equality we have shown that the distance is indeed a metric and the proof ends.

\end{proof}

%%%%%%%%%%%%%%%%%%%%%%%%%%%%%%%%%%%%%%%%%%%%%%%%%%%%%%%%%%%%%%%%%%%%%%%%%%%%%%%%%%%%%%%%%%%%%%%%%%%%%%%%%%%%%%%%%%%%%%%%%
%
%LOCALIZED MODULUS METRIC
%
%%%%%%%%%%%%%%%%%%%%%%%%%%%%%%%%%%%%%%%%%%%%%%%%%%%%%%%%%%%%%%%%%%%%%%%%%%%%%%%%%%%%%%%%%%%%%%%%%%%%%%%%%%%%%%%%%%%%%%%%%

\section{Analysis of the localized modulus-metric}\label{section_analysis_localized_modulus_metric}

%%%%%%%%%%%%%%%%%%%%%%%%%%%%%%%%%%%%%%%%%%%%%%%%%%%%%%%%%%%%%%%%%%%%%%%%%%%%%%%%%%%%%%%%%%%%%%%%%%%%%%%%%%%%%%%%%%%%%%%%%
%
%PROPOSITION 1 - LOCALIZED MODULUS METRIC
%
%%%%%%%%%%%%%%%%%%%%%%%%%%%%%%%%%%%%%%%%%%%%%%%%%%%%%%%%%%%%%%%%%%%%%%%%%%%%%%%%%%%%%%%%%%%%%%%%%%%%%%%%%%%%%%%%%%%%%%%%%
\begin{proposition}
$d_n(T,\bar{T}) < \infty$.
\end{proposition}
\begin{proof}
From Eq. \ref{equation_hausdorff_ab}, we have, $\forall s \in [a,b]$,
\begin{equation}
 h(T,\bar{T}) \ge \left| d(s,T) - d(s,\bar{T}) \right|.
\end{equation}
Multiplying by $\mathcal{L}(\cdot)$, which is positive, we obtain
\begin{equation}
h(T,\bar{T}) \; \mathcal{L}(b-s) \ge \mathcal{L}(b-s) \; \left| d(s,T) - d(s,\bar{T}) \right|.
\end{equation}
By integrating the above equation, we obtain
\begin{equation}
h(T,\bar{T}) \int_{a}^{b} \mathcal{L}(b-s) \; \textrm{d}s \ge \int_{a}^{b} \mathcal{L}(b-s) \; \left| d(s,T) - d(s,\bar{T}) \right| \; \textrm{d}s.
\end{equation}
Since $h(T,\bar{T}) \le b-a$ and $\mathcal{L}(x) < m$, $\forall x \in [0,b-a]$ it follows that
\begin{equation}
 d_n(T,\bar{T}) < m \; (b-a)^2 < \infty.
\end{equation}
\end{proof}

%%%%%%%%%%%%%%%%%%%%%%%%%%%%%%%%%%%%%%%%%%%%%%%%%%%%%%%%%%%%%%%%%%%%%%%%%%%%%%%%%%%%%%%%%%%%%%%%%%%%%%%%%%%%%%%%%%%%%%%%%
%
%PROPOSITION 2 - LOCALIZED MODULUS METRIC
%
%%%%%%%%%%%%%%%%%%%%%%%%%%%%%%%%%%%%%%%%%%%%%%%%%%%%%%%%%%%%%%%%%%%%%%%%%%%%%%%%%%%%%%%%%%%%%%%%%%%%%%%%%%%%%%%%%%%%%%%%%
\begin{proposition}
$d_n : \mathcal{S}_{[a, b]} \times \mathcal{S}_{[a, b]} \to \mathbb{R}^+$ is a metric.
\end{proposition}
\begin{proof}
In order to show that $d_n(\cdot,\cdot)$ is a metric we need to prove that it is non-negative, that $d_n(T,\bar{T}) = 0 \Leftrightarrow T = \bar{T}$ for any $T,\bar{T}$, that it is symmetric, and that it satisfies the triangle inequality.

Let $T,\bar{T} \in \mathcal{S}_{[a, b]}$. It is trivial to show that $d_n(T,\bar{T})$ is non-negative and symmetric, and that $T = \bar{T} \Rightarrow d_l(T,\bar{T}) = 0$. In order to prove the converse we use a \textit{reductio ad absurdum} argument. Assume $T \neq \bar{T}$ with $d_n(T,\bar{T}) = 0$. For $s \in [a,b]$ let
\begin{equation}
q(s) = \mathcal{L}(b-s) \; \left|d(s,T) - d(s,\bar{T})\right|.
\end{equation}
Because $\left|d(s,T) - d(s,\bar{T})\right|$ is continuous and $\mathcal{L}$ is continuous it results that $q$ is continuous. Because $T \neq \bar{T} \Rightarrow \exists \;  s \in [a,b]$ such that $d(s,T) \neq d(s,\bar{T})$; because $\mathcal{L}$ is strictly positive on $(0, b-a]$, it follows that $q$ is not zero everywhere, $\exists s \in [a,b)$ such that $q(s) > 0$. Because $q$ is continuous, it follows that
\begin{equation}
d_n(T,\bar{T}) = \int_{a}^{b} q(s) \; \textrm{d}s > 0
\end{equation}
which contradicts the hypothesis that $d_n(T,\bar{T}) = 0$. Hence, $T = \bar{T}$.

In order to prove the triangle inequality consider $\hat{T} \in \mathcal{S}_{[a, b]}$. We have, $\forall s \in \mathbb{R}$,
\begin{equation}
| d(s,T) - d(s,\hat{T}) + d(s,\hat{T}) -d(s,\bar{T})| \leq  |d(s,T) - d(s,\hat{T}) | + |d(s,\hat{T}) -d(s,\bar{T})|.
\end{equation}
Multiplying by $\mathcal{L}(\cdot)$ and integrating the above equation, we obtain
\begin{align}
&\int_a^b  \left|d(s,T) - d(s,\bar{T})\right| \; \mathcal{L}(b-s) \; \textrm{d}s   \le \nonumber \\
&\int_a^b  \left|d(s,T) - d(s,\hat{T})\right| \; \mathcal{L}(b-s) \; \textrm{d}s +
\int_a^b   \left|d(s,\hat{T}) - d(s,\bar{T})\right| \; \mathcal{L}(b-s) \; \textrm{d}s
\end{align}
It results that
\begin{equation}
d_n(T,\bar{T}) \le d_n(T,\hat{T}) + d_n(\hat{T},\bar{T}).
\end{equation}
With this final inequality we have shown that the distance is indeed a metric and the proof ends.
\end{proof}

\section{Algorithms}

\SetAlFnt{\small}
\SetAlCapSkip{1ex}
\SetAlgoInsideSkip{smallskip}
\SetCommentSty{textrm}
\SetKwComment{Comment}{}{}
\SetKwFunction{and}{and}
\SetKwFunction{length}{length}
\SetKwFunction{sort}{sort}

\begin{algorithm}[ph]
\KwIn{The pair of nonempty spike trains $T_1$, $T_2$ and the bounds $a$ and $b$.}
\KwOut{The distance  $d_o (T_1, T_2)$ between the spike trains.}
\Comment{$\langle$$T_1$, $T_2$ and $P$ are ordered sets of real numbers, indexed starting from 0.$\rangle$}
$n_1:=\length(T_1); \; n_2:=\length(T_2)$\;
$P:=\text{sort}(T_1 \; \bigcup \; T_2$)\;
$M:=\emptyset$\;
\For{$i:= 1 \ldots n_1-1$} {
    $M:= M \; \bigcup \; \{(T_1[i]+T_1[i-1])/2\}$\;
}
\For{$i:= 1 \ldots n_2-1$} {
    $M:= M \; \bigcup \; \{(T_2[i]+T_2[i-1])/2\}$\;
}
\For{$i:= 1 \ldots \length(P)-1$} {
    $M:= M \; \bigcup \; \{(P[i]+P[i-1])/2\}$\;
}
$P:=\text{sort}(P \; \bigcup \; M  \; \bigcup \; \{a, b\}$)\;
$d_o:=0; \; i_1:=0; \; i_2:=0$\;
\Comment{$\langle$$s$ is the currently considered point from $P$. $\phi$ is the value at $s$ of the integrated function $|d(s,T_1) - d(s, T_2)|$. $s_p$ is the previously considered point from $P$, and $\phi_p$ is the value at $s_p$ of the integrated function. $i_1$ is the index of the first spike in $T_1$ having a timing that is greater than $s$, if there is such a spike, or the index of the last spike of $T_1$, otherwise. $i_2$ is computed analogously for $T_2$.$\rangle$}
$s_p:=a; \; \phi_p:=|T_1[0]-T_2[0]|$\;
 \For{$i:= 1 \ldots \length(P)-1$}{
    $s:=P[i]$\;
    \While{$s \geq T_1[i_1] \; \and \; i_1<n_1-1$}{
        $i_1:=i_1+1$\;
    }
    \While{$s \geq T_2[i_2] \; \and \; i_2<n_2-1$}{
        $i_2:=i_2+1$\;
    }
    $d_1:=b-a; \; d_2:=b-a$\;
    \If{$i_1>0$}{
        $d_1:=s-T_1[i_1-1]$\;
    }
    $d_1'=|T_1[i_1]-s|$\;
    \If{$d_1'<d_1$}{
        $d_1:=d_1'$\;
    }
    \If{$i_2>0$}{
        $d_2:=s-T_2[i_2-1]$\;
    }
    $d_2'=|T_2[i_2]-s|$\;
    \If{$d_2'<d_2$}{
        $d_2:=d_2'$\;
    }
    \Comment{$\langle$We now have $d_1=d(s, T_1)$ and $d_2=d(s, T_2)$ and we can compute the value of $\phi$ at $s$:$\rangle$}
    $\phi:=|d_1-d_2|$\;
    \Comment{$\langle$The integration is performed here:$\rangle$}
    $d_o:=d_o+(s-s_p)(\phi+\phi_p)/2$\;
     $s_p:=s; \; \phi_p:=\phi$\;
 }
\Return $d_o$\;
\caption{An algorithm for computing the modulus-metric distance $d_o$ between two spike trains $T_1$ and $T_2$. The text surrounded by  $\langle\ldots\rangle$  represents comments.
}
\label{algorithm1}

\end{algorithm}

%-------------------------------

\begin{algorithm}[ph]
\KwIn{The pair of non-empty spike trains $T_1$, $T_2$ and the bounds $a$ and $b$.}
\KwOut{The distance  $d_o (T_1, T_2)$ between the spike trains.}
\Comment{$\langle$$T_1$ and $T_2$ are ordered sets of real numbers, indexed starting from 0. $i_1$ and $i_2$ are the indices of the currently processed spikes in the two spike trains.  $p_1$ and $p_2$ are the indices of the previously processed spikes in the two spike trains. $p$ is the index of the spike train to which the previously processed spike belonged (1 or 2), after at least one spike has been processed, or 0 otherwise.$\rangle$}
$i_1:=0; i_2:=0; p_1:=0; p_2:=0; p:=0$\;
$n_1:=\length(T_1); \; n_2:=\length(T_2)$\;
\Comment{$\langle$$P$ is an array of structures $(s, \phi)$ consisting of an ordered pair of numbers.$\rangle$}
$P:=\left\{\left( s \mapsto a, \; \phi \mapsto \left|T_1[0]-T_2[0]\right| \right) \right\}$\;
\Comment{$\langle$Process the spikes until the end of one of the spike trains is reached.$\rangle$}
\While{$i_1 < n_1 \; \and \; i_2 < n_2$}{
    \If{$T_1[i_1] \leq T_2[i_2]$}{
       proc1(1, 2)\;
    }
    \Else{
        proc1(2, 1)\;
    }
}
\Comment{$\langle$Process the rest of the spikes in the spike train that has not been fully processed.$\rangle$}
\While{$i_1 < n_1$}{
    proc2(1, 2)\;
}
\While{$i_2 < n_2$}{
    proc2(2, 1)\;
}
$P:=P \; \bigcup \; \left\{ \left( s \mapsto b, \; \phi \mapsto \left|T_1[n_1-1]-T_2[n_2-1]\right| \right) \right\}$\;
\Comment{$\langle$Sort $P$. Elements of $P$ are sorted according to their value of $s$.$\rangle$}
$P:=\text{sort}(P)$\;
\Comment{$\langle$Perform the integration.$\rangle$}
$d_o:=0$\;
\For{$i:= 1 \ldots \length(P)-1$}{
    $d_o:=d_o+(P[i].s-P[i-1].s)(P[i].\phi+P[i-1].\phi)/2$\;
 }
 \Return $d_o$\;
 \caption{Another algorithm for computing the modulus-metric distance $d_o$ between two spike trains $T_1$ and $T_2$.
The text surrounded by  $\langle\ldots\rangle$  represents comments.
 }
 \label{algorithm2}
\end{algorithm}

\begin{procedure}[ph]
\KwIn{The indices $j$ and $k$ of the two sorted spike trains $T_j$, $T_k$; $j, k \in \{1, 2\}$, $j \neq k$.}
 \KwData{Uses as global variables: the indices $i_j$,  $i_k$ of the current spikes in the two spike trains; the indices $p_j$,  $p_k$ of the previously processed spikes in the two spike trains; the index $p$ of the spike train to which the previously processed spike belonged (1 or 2; if no spike has been previously processed, $p=0$); the data structure $P$. We should have here $T_j[i_j] \leq T_k[i_k]$. If $p \neq 0$, we should have $T_p[i_p] \leq T_j[i_j]$.}
\KwResult{Performs part of the processing needed for creating $P$. The procedure is used when processing has not reached the end of one of the spike trains.}
\If{$i_j>0$}{
    \Comment{$\langle$Adds to $P$ the timing situated at the middle of the interval between the currently processed spike and the previous spike in the same spike train.$\rangle$}
    $t:=(T_j[i_j]+T_j[i_j-1])/2$\;
    \Comment{$\langle$We have $d(t, T_j)=T_j[i_j]-t=t-T_j[i_j-1]=(T_j[i_j]-T_j[i_j-1])/2$.$\rangle$}
    $P:=P \; \bigcup \; \left\{\left( s \mapsto t, \; \phi\mapsto \left|(T_j[i_j]-T_j[i_j-1])/2)-d(t, k, i_k)\right| \right)\right\} $\;
}
\If{$p = k$}{
    \Comment{$\langle$If the previously processed spike was one from the other spike train than the spike currently processed, adds to $P$ the timing situated at the middle of the interval between the currently processed spike and the previously processed spike.$\rangle$}
    $t:=(T_j[i_j]+T_k[p_k])/2$\;
    \Comment{$\langle$Since $t$ is at equal distance to the closest spikes in the two spike trains, $T_j[i_j]$ and $T_k[p_k]$, we have $d(t, T_j)=d(t, T_k)$ and $\phi(t)=0$.$\rangle$}
    $P:=P \; \bigcup \; \left\{\left( s \mapsto t, \; \phi \mapsto 0 \right) \right\} $\;
}
\Comment{$\langle$Adds to $P$ the currently processed spike.$\rangle$}
$t:=T_j[i_j]$\;
\Comment{$\langle$We have $d(t, T_j)=0$. If at least one spike from $T_k$ has been processed, we have $T_k[p_k] \leq t \leq T_k[i_k]$, with $i_k=p_k+1$, and thus $d(t, T_k)=\min(|t-T_k[p_k]|, \; T_k[i_k]-t)$. If no spike from $T_k$ has been processed, we have $p_k=i_k=0$, and the previous formula for $d(t, T_k)$ still holds.$\rangle$}
$P:=P \; \bigcup \; \left\{\left( s \mapsto t, \; \phi \mapsto \min(|t-T_k[p_k]|, \; T_k[i_k]-t)\right)\right\} $\;
$p_j:=i_j$\;
$i_j:=i_j+1$\;
$p:=j$\;
\caption{proc1($j, k$).}
\end{procedure}

\begin{procedure}[ph]
\KwIn{The indices $j$ and $k$ of the two sorted spike trains $T_j$, $T_k$; $j, k \in \{1, 2\}$, $j \neq k$.}
\KwData{Uses as global variables: the index $i_j$ of the current spike in $T_j$; the index $p_k$ of the previously processed spike in $T_k$; the index $p$ of the spike train to which the previously processed spike belonged (1 or 2); the data structure $P$. Here, $p_k$ should be the index of the last spike in spike train $T_k$. We should have $T_k[p_k] \leq T_j[i_j]$.}
\KwResult{Performs part of the processing needed for creating $P$. The procedure is used when processing has reached the end of spike train $T_k$.}
\If{$i_j>0$}{
    \Comment{$\langle$Adds to $P$ the timing situated at the middle of the interval between the currently processed spike and the previous spike in the same spike train.$\rangle$}
    $t:=(T_j[i_j]+T_j[i_j-1])/2$\;
    \Comment{$\langle$We have $d(t, T_j)=T_j[i_j]-t=t-T_j[i_j-1]=(T_j[i_j]-T_j[i_j-1])/2$.$\rangle$}
    $P:=P \; \bigcup \; \left\{\left( s \mapsto t, \; \phi \mapsto \left|(T_j[i_j]-T_j[i_j-1])/2)-d(t, k, p_k)\right| \right) \right\}$\;
}
\If{$p = k$}{
    \Comment{$\langle$If the previously processed spike was one from the other spike train than the spike currently processed (i.e., the last spike in the spike train that has been fully processed), adds to $P$ the timing situated at the middle of the interval between the currently processed spike and the previously processed spike.$\rangle$}
    $t:=(T_j[i_j]+T_k[p_k])/2$\;
    \Comment{$\langle$Since $t$ is at equal distance to the closest spikes in the two spike trains, $T_j[i_j]$ and $T_k[p_k]$, we have $d(t, T_j)=d(t, T_k)$ and $\phi(t)=0$.$\rangle$ }
     $P:=P \; \bigcup \; \left\{\left( s \mapsto t, \; \phi \mapsto 0 \right)\right\} $\;
}
\Comment{ $\langle$ Adds to $P$ the currently processed spike. $\rangle$ }
$t:=T_j[i_j]$\;
\Comment{ $\langle$ We have $d(t, T_j)=0$. We have $T_k[p_k] \leq t$ and the spike at $p_k$ is the last one in $T_k$, and thus $d(t, T_k)=t-T_k[p_k]$.$\rangle$}
$P:=P \; \bigcup \; \left\{\left( s \mapsto t, \; \phi \mapsto t-T_k[p_k]\right) \right\}$\;
%$p_j:=i_j$\;
$i_j:=i_j+1$\;
$p:=j$\;
\caption{proc2($j, k$).}
\end{procedure}

\begin{function}[h!t]
\KwIn{A timing $t$, the index $k \in \{1, 2\}$ of a sorted spike train $T_k$,  and the index $i$ of a spike in $T_k$, such that either $t \leq T_k[i]$ or $i$ is the index of the last spike of $T_k$.}
\KwOut{The distance  $d(t, T_k)$ between the timing $t$ and the spike train $T_k$.}
$d:=|T_k[i]-t|$\;
$j:=i-1$\;
\While{$j \geq 0 \; \and \; |T_k[j]-t| \leq d$}{
        $d:=|T_k[j]-t|$\;
        $j:=j-1$\;
    }
\Return $d$\;
\caption{d($t, k, i$).}
\end{function}

\newpage
\bibliography{metrica}

\end{document}